\documentclass[12pt,twoside]{article}

\def\myPrimaryAuthor{A.Itkin, P.Carr}
\def\myAuthor{
Andrey Itkin \thanks{Hap Capital LLC \& Department of Mathematics, Rutgers University, New Jersey},
\ Peter Carr \thanks{Bloomberg LP \& New York University}}
\def\myProjectName{Using pseudo-parabolic and fractional equations for option pricing in jump diffusion models}
\def\myProjectNum{}
\def\myAbsract{In mathematical finance a popular approach for pricing options under some L\'evy model is to consider underlying that follows a Poisson jump diffusion process. As it is well known this results in a partial integro-differential equation (PIDE) that usually does not allow an analytical solution while numerical solution brings some problems. In this paper we elaborate a new approach on how to transform the PIDE to some class of so-called pseudo-parabolic equations which are known in mathematics but are relatively new for mathematical finance. As an example we discuss several jump-diffusion models which L\'evy measure allows such a transformation.}

\usepackage[dvips]{graphicx}
\usepackage{amsmath, amsthm, amssymb, epsfig}
\usepackage{dsfont}
\usepackage[usenames,dvipsnames]{color}
\usepackage[english]{babel}
\usepackage{fancyhdr}
\usepackage{verbatim}
\usepackage{caption}
\usepackage[comma, sort&compress, authoryear, round]{natbib}
\usepackage{float}

\def\ni{\noindent}

\def\bi{\begin{itemize}}     \def\ei{\end{itemize}}
\def\ben{\begin{enumerate}}   \def\een{\end{enumerate}}
\def\beq{\begin{equation}}   \def\eeq{\end{equation}}
\def\bd{\begin{displaymath}} \def\ed{\end{displaymath}}
\def\bea{\begin{eqnarray}}   \def\beaz{\begin{eqnarray*}}
\def\eea{\end{eqnarray}}     \def\eeaz{\end{eqnarray*}}
\def\beac{\begin{eqnarrayc}} \def\eeac{\end{eqnarrayc}}
\def\nn{\nonumber}           
\def\dsize{\displaystyle}    \def\dfrac#1#2{\frac{\dsize #1}{\dsize #2}}

\newtheorem{theorem}{Theorem}[section]
\newtheorem{proposition}[theorem]{Proposition}

\newcommand{\bc}{\begin{center}}
\newcommand{\ec}{\end{center}}
\newcommand{\be}{\begin{equation}}
\newcommand{\ee}{\end{equation}}
\newcommand{\bq}{\begin{eqnarray}}
\newcommand{\eq}{\end{eqnarray}}

\newcommand{\fp}[2]{\frac{\partial #1}{\partial #2}}
\newcommand{\sop}[2]{\frac{\partial^2 #1}{\partial #2^2}}

\newcommand{\cp}[3]{\frac{\partial^2 #1}{\partial #2 \partial #3}}

{\catcode`$=9 $ 
}

   {\noindent \footnotesize \verbatim}%
   {\normalsize \endverbatim}

\normalsize

\begin{document}

\pagestyle{plain}
\vspace{-2cm}
\begin{figure}
\begin{minipage}[t]{\textwidth}
\DeclareGraphicsRule{.jpg}{bmp}{.bb}{}
\hspace{-3mm}
\end{minipage}
\end{figure}

\title{\vspace{2cm}\color{blue}\myProjectName \break \break \myProjectNum}
\author{\color{red}\myAuthor}
\date{}
\maketitle

\begin{center}
{\color{magenta} Submitted to Applied Mathematical Finance} \\
\end{center}
\smallskip

\renewcommand\abstractname{Abstract}
\abstract{\myAbsract}

\newpage
\pagestyle{fancyplain}
\fancyfoot[RE,RO]{\thepage}
\fancyfoot[LE,LO]{\color{red}\myPrimaryAuthor\hfill}
\fancyfoot[CE,CO]{\color{red} \footnotesize{\myProjectName \\
\myProjectNum}}

\def\itemDesc#1#2#3{
    \colorbox{#1} {\parbox{0.9\columnwidth}{\color{#2} {\scriptsize #3}}\\}
}
\setcounter{tocdepth}{3}
\tableofcontents
\newpage

\section{Introduction}
In mathematical finance a popular approach for pricing options under some L\'evy model is to consider underlying that follows a Poisson jump diffusion process. As it is well known this results in a partial integro-differential equation (PIDE) that usually does not allow an analytical solution while numerical solution brings some problems.
These problems are mainly related to computing a non-local integral term while we assume that computing a differential part of the PIDE, being discussed numerous times in the literature, could be provided in a relatively standard way. Moreover, using splitting technique it is always possible to reduce the whole PIDE to a series of equations part of which are pure PDE and the remaining part are pure evolutionary-integral equations (EIDE) (see, for instance, \cite{HoutWelfert2009, ItkinCarrBarrierR3}). Thus, further on we will consider just the later.
A thorough description of methods used for solving this kind of equation is given in \cite{ContBook2009, Hilber2009} while problems related to implementation of these methods are discussed in \cite{CarrMayo, Strauss2006}.

According to the last cited paper we could distinguish the following methods that were used to solve the EIDE. In an early paper, Amim (1993) used an explicit multinomial tree based approach. D'Halluin et al. (2004, 2005b) implemented implicit methods for evaluating vanilla European options, barrier options, and American options. They also showed that when a log spaced grid is used with a Crank Nicolson discretization on a problem with constant parameters the resulting scheme is unconditionally strictly stable. In addition, they showed that the simple Picard iteration scheme (also suggested by Tavella and Randall (2000)) for solving the discretized equations is globally convergent. Specifically, they reported that when they priced options in the Merton model the error was reduced by two orders of magnitude at each iteration for typical values of the time step size and Poisson arrival intensity. More recently, d'Halluin et al. (2005a) presented a semi-Lagrangian approach for pricing American Asian options under jump diffusion processes. Andersen and Andreasen (2000) derived a forward equation describing the evolution of European call options as functions of strike and maturity, and discussed its application to the problem of fitting the stock process to option prices in the market. They also presented a second order accurate unconditionally stable operator splitting (ADI) method for pricing options which does not require iterative solution of an algebraic equation at each time step. (Unfortunately, it is not clear how to extend their method to the valuation of American options while retaining second order accuracy.) Cont and Voltchkova (2003) used a discretization that is implicit in the differential terms and implicit in the integral term, and showed that it converges to a viscosity solution. Their method extends to infinite activity models, and does not require the diffusion part of the equation to be non-degenerate. These partial integro-differential equations have also been solved by many others. See, for example, Zhang (1993) and Matache et al. (2002). Although the pricing equations have often been solved numerically, because of the integrals in the equations the methods have proven relatively expensive. The obvious discretizations of the pricing equations combine standard discretization methods for the differential terms with quadrature methods such as Simpson's rule or Gaussian quadrature for evaluating the integral term. This approach is computationally expensive since the integral must be approximated at each point of the mesh used for discretizing the differential terms. The difficulties are greater if an implicit discretization of both the integral and the differential terms is used. The expense of evaluating the integral at all points of the computational grid can, however, be reduced by making the same exponential change of variables often used when solving the Black–Scholes differential equation when there is no jump process. This converts the integral term into a correlation integral which can be evaluated at all the mesh points simultaneously using the Fast Fourier Transform. This approach has been suggested by many authors (Wilmott, 1998; Tavella and Randall, 2000; Andreasen and Anderson, 2000).

For multidimensional L{\'e}vy process various kind of finite elements methods were proposed (see survey in \cite{Hilber2009}) because finite difference methods are not efficient when dimensionality of the problem exceeds 3.

Also note a new method for exponential jumps proposed by \cite{LiptonSepp2009} who calculate the jump integral recursively on the spatial grid. This is a special trick for exponential jumps, it does not work for more familiar Gaussian jumps. The authors claim that for discrete jumps a simple interpolation routine is sufficient.

As Carr and Mayo mentioned (see \cite{CarrMayo}) quadrature methods are expensive since the integrals must be evaluated at every point of the mesh. Though less so, Fourier methods are also computationally intensive since in order to avoid wrap around effects they require enlargement of the computational domain. They are also slow to converge when the parameters of the jump process are not smooth, and for efficiency require uniform meshes. Therefore, they proposed a different and more efficient class of methods which are based on the fact that the integrals often satisfy differential equations. Depending on the process the asset follows, the equations are either ordinary differential equations or parabolic partial differential equations. Both types of equations can be accurately solved very rapidly. They used to demonstrate the advantage of such an approach for the Merton and Kou models. However, for other types of the L\'evy models an extension of their idea is unknown yet.

Therefore in this paper we propose two different approaches. The idea of the first one is to represent a L\'evy measure as the Green's function of some yet unknown differential operator $\mathcal{A}$. If we manage to find an explicit form of such an operator then the original PIDE reduces to a new type of equation - so-called pseudo-parabolic equation. These equations are known in mathematics (see, for instance, \cite{CannonLin2007}) but are new for mathematical finance.

Then we rely on two important results, namely: a) the inverse operator $\mathcal{A}^{-1}$ exists, and b) the obtained pseudo parabolic equation could be formally solved analytically via a matrix exponent. Having that we discuss a numerical method of how to compute this matrix exponent. We show that we can do it using a finite difference scheme similar to that used for solving parabolic PDEs and the matrix of this FD scheme is banded. We fulfill this program for general tempered stable processes (GTSP) with an integer damping exponent $\alpha$.

Alternatively for some class of L\'evy processes, known as GTSP/KoBoL/SSM models, with the real dumping exponent $\alpha$ we show how to transform the corresponding PIDE to a fractional PDE (method 2). Fractional PDEs for the L\'evy processes with finite variation were derived by \cite{BL2002}  and later by \cite{Cartea2007}. using a characteristic function technique. Numerical solution of these equations was investigated by \cite{Cartea2007} and \cite{Marom2009}. In this paper we derive them in all cases including processes with infinite variation using a different technique - shift operators. Then to solve them we apply a new method, namely: having results computed for $\alpha \in \mathbb{I}$ we then interpolate them with the second order in $\alpha$ to obtain the solution at any $\alpha \in \mathbb{R}$.

We also show that despite it is a common practice to integrate out all L\'evy compensators in the integral term when one considers jumps with finite activity and finite variation, this breaks the stability of the scheme, at least for the fractional PDE. Therefore, in order to construct the unconditionally stable scheme one must keep the other terms under the integrals. To resolve this in Cartea (2007) the authors were compelled to change their definition of the fractional derivative.

We also propose the idea of solving FPDE with real $\alpha$ by using interpolation between option prices computed for the closest integer values of $\alpha$. For the latter an efficient scheme is proposed that results in LU factorization of the band matrix.

It is important to note that both proposed methods could be easily generalized for a time-dependent L\'evy density.

The rest of the paper is organized as follows. In section \ref{Sbasic} we discuss a basic example of the method which is built based on a simple exponential L\'evy measure. In the next section we consider GTSP models and show how to reduce the corresponding PIDE to a pseudo parabolic equation in this case. Section \ref{nm1} describes numerical solution of the obtained pseudo parabolic equations in case $\alpha \in \mathbb{I}$. Section \ref{Sgc} describes a general case of real $\alpha$ and introduces our method of deriving fractional PDE based on shift operators. In section \ref{Snum} we discuss how to solve these FPDE by constructing unconditionally stable finite difference schemes of high order of accuracy in space and time and provide some numerical examples and comparison with the other methods. The last section concludes.

\section{Basic model} \label{Sbasic}
In this section we consider the simplest possible problem to demonstrate basics of our new method. We assume no arbitrage so that there exists a risk-neutral measure $\mathbb{Q}$. We assume zero interest rates and dividends so that the stock price is a $\mathbb{Q}$-martingale. Suppose that the underlying stock price process is pure jump (i.e. there is no continuous martingale component). Further suppose that the jump process is a compound Poisson process. The arrival rate of a jump is constant at $\lambda >0$, while the jump size distribution is a symmetric Laplace distribution, i.e the probability density for a jump of size $j \in \mathbb{R}$, given that a jump has occurred is given by:
\be
q(j) = \frac{\alpha e^{- \alpha |j|}}{2}, \qquad j \in \mathbb{R},
\label{jpdf}
\ee
where $\alpha > 0$ is a free parameter.
We recognize that these dynamics let prices become negative and ignore this complication.
Let $u(x,t)$ be the value of the contingent claim at calendar time $t \in [0,T]$
given that the time $t$ stock price is $x \in \mathbb{R}$.
As a result of our assumptions,
the contingent claim value solves the following PIDE:
\be
\fp{}{t} u(x,t) + \lambda
\int_{\mathbb{R}} [u(x+j,t) - u(x,t) - \fp{}{x}u(x,t) j] q(j) dj = 0,
\label{pide}
\ee
on the domain $x \in \mathbb{R}, t \in [0,T]$.
For a European call, the terminal condition is
\be
u(x,T) = (x - K)^+, \qquad x \in \mathbb{R},
\label{ec}
\ee
where $K \in \mathbb{R}$ is the strike price.

Now the symmetry of the PDF in (\ref{jpdf}) implies that:
\be
\int_{\mathbb{R}} j q(j) dj = 0,
\label{sym}
\ee
and hence the PIDE (\ref{pide}) simplifies to:
\be
\fp{}{t} u(x,t) - \lambda u(x,t) + \lambda
\int_{\mathbb{R}} u(x+j,t) q(j) dj = 0.
\label{pide1}
\ee
If we do the change of variable $z=- j$ in the integral, we obtain a convolution:
\be
\fp{}{t} u(x,t) - \lambda u(x,t) + \lambda \int_{\mathbb{R}} u(x-z,t) q(z) dz = 0.
\label{pide2}
\ee
If we do the change of variable $y= x-z$ in the integral, we obtain:
\be
\fp{}{t} u(x,t) - \lambda u(x,t) + \lambda \int_{\mathbb{R}} u(y,t) q(x-y) dy = 0.
\label{pide3}
\ee

Now consider the simple second order linear inhomogeneous ODE:
\be
g''(x) - \alpha^2 g(x) = - \delta(x), \qquad x \in \mathbb{R},
\label{ode}
\ee
where $\delta(x)$ denotes Dirac's delta function.
Suppose that the ODE is to be solved subject to the boundary conditions:
\be
\lim\limits_{x \rightarrow \pm \infty} g(x) = 0.
\label{bc}
\ee
The solution to this problem is usually referred to as a
Green's function. The solution is well known to be:
\be
g(x) = \frac{e^{- \alpha |x|}}{2 \alpha}.
\label{odesoln}
\ee
Comparing (\ref{odesoln}) and (\ref{jpdf}), we see that:
\be
q(x) = \alpha^2 g(x).
\label{q=g}
\ee
Hence, the PIDE (\ref{pide5}) can be re-written as:
\be
\fp{}{t} u(x,t) - \lambda u(x,t) + \lambda \alpha^2
\int_{\mathbb{R}} u(y,t) g(x-y) dy = 0.
\label{pide6}
\ee

To exploit the connection (\ref{q=g}), let ${\cal D}_x$ denote the first derivative operator
and let ${\cal A}_x$ denote the following linear differential operator:
\be
{\cal A}_x \equiv {\cal D}_x^2  - \alpha^2 {\cal I}_x,
\label{op}
\ee
where ${\cal I}_x$ is the identity operator.

Using this operator notation, the ODE (\ref{ode}) reads:
\be
{\cal A}_x  g(x) = - \delta(x).
\label{ode1}
\ee
Suppose that we apply the ${\cal A}_x$ operator to (\ref{pide6}):
\be
{\cal A}_x \fp{}{t} u(x,t) - \lambda {\cal A}_x u(x,t) + \lambda \alpha^2
\int_{\mathbb{R}} u(y,t) {\cal A}_x g(x-y) dy = 0.
\label{pide4}
\ee
where we have assumed that the interchange of the integral and the
differential operator is permissible.
Substituting (\ref{ode1}) in (\ref{pide6}) implies that:
\be
{\cal A}_x \fp{}{t} u(x,t) - \lambda {\cal A}_x u(x,t) - \lambda \alpha^2
\int_{\mathbb{R}} u(y,t) \delta(x-y) dy = 0.
\label{pide5}
\ee
Using the sifting property of the delta function implies
that our problem reduces to a (third order) PDE:
\be
{\cal A}_x \fp{}{t} u(x,t) - \lambda {\cal A}_x u(x,t) - \lambda \alpha^2
 u(x,t)  = 0.
\label{pide7}
\ee
Substituting (\ref{op}) in  (\ref{pide7}) and simplifying implies:
\be
\frac{\partial^3 }{\partial x^2 \partial t} u(x,t) - \alpha^2 \fp{}{t} u(x,t)
- \lambda \sop{}{x} u(x,t) = 0.
\label{simp}
\ee

Note that the generalization from exponential type kernels to Erlang type kernels can be handled by replacing the
second order differential operator ${\cal A}_x$ by a higher order differential operator.
We further note that the Central Limit Theorem implies that the limiting sum of these independent exponential random variables is normally distributed. A Gaussian component to the jump kernel induces an infinite order ODE
which is equivalent to a PDE. Hence the Gaussian type jump of Merton can be handled by solving a PDE as we already know. The PDF of a linear combination of independent exponential and Gaussian random variables is called the Polya Laguerre distribution. A good reference for the above inversion is Hirschman and Widder.

\section{GTSP/KoBoL/SSM model}
Stochastic skew model (SSM) has been proposed by \cite{CarrWu2004} for pricing currency options. It makes use of a L\'evy model also known as generalized tempered stable processes (GTSP) (see \cite{ContTankov}) for the dynamics of stock prices which generalize the CGMY processes proposed by \cite{CarrGemanMadanYor:2002}. A similar model was independently proposed by \cite{Koponen} and then \cite{BL2002} The processes are obtained by specifying a more generalized L\'evy measure with two additional parameters. These two parameters provide control on asymmetry of small jumps and different frequencies for upward and downward jumps. The results of \cite{Hagan2006} show that this generalization allows for more accurate pricing of options.

Generalized Tempered Stable Processes (GTSP) have probability densities symmetric in a neighborhood of the origin and exponentially decaying in the far tails. After this exponential softening, the small jumps keep their initial stable-like behavior, whereas the large jumps become exponentially tempered. The L\'evy measure of GTSP reads
\begin{equation}\label{measure}
    \mu(y) = \lambda_{-} \dfrac{e^{-\nu_{-}|y|}}{|y|^{1 + \alpha_{-}}}\mathbf{1}_{y<0} + \lambda_{+} \dfrac{e^{-\nu_{+}|y|}}{|y|^{1 + \alpha_{+}}}\mathbf{1}_{y>0},
\end{equation}

\ni where $\nu_\pm > 0, \lambda_\pm > 0$ and $\alpha_\pm < 2$. The last condition is necessary to provide
\begin{equation}\label{cond}
    \int^1_{-1} y^2 \mu(dy) < \infty, \ \int_{|y| > 1} \mu(dy) < \infty.
\end{equation}

The case $\lambda_{+} = \lambda_{-}, \alpha_{+} = \alpha_{-}$ corresponds to the CGMY process. The limiting case $\alpha_{+} = \alpha_{-} = 0, \lambda_{+} = \lambda_{-}$ is the special case of the Variance Gamma process of \cite{MadanSeneta:90}. As Hagan at al mentioned (see \cite{Hagan2006}) six parameters of the model play an important role in capturing various aspects of the stochastic process. The parameters $\lambda_\pm$ determine the overall and relative frequencies of upward and downward jumps. If we are interested only in jumps larger than a given value, these
two parameters tell us how often we should expect such events. $\nu_\pm$ control the tail behavior of the L\'evy measure, and they tell us how far the process may jump. They also lead to skewed distributions when they are unequal. In the special case when they are equal, the L\'evy measure is symmetric. Finally, $\alpha_\pm$ are particularly useful for the local behavior of the process. They determine whether the process has finite or infinite activity, or variation.

Using this model of jumps \cite{CarrWu2004} derived the following PIDE which governs an arbitrage-free value of a European call option at time $t$
\begin{align} \label{pide_init}
r_d C(&S,V_R,V_L,t)  =  \fp{}{t} C(S,V_R,V_L,t) + (r_d - r_f) S \fp{}{S} C(S,V_R,V_L,t)  \\
&+  \kappa(1 - V_R) \fp{}{V_R} C(S,V_R,V_L,t) + \kappa(1 - V_L) \fp{}{V_L} C(S,V_R,V_L,t) \nonumber \\
&+ \frac{\sigma^2 S^2 (V_R+V_L)}{2} \sop{}{S} C(S,V_R,V_L,t) + \sigma \rho^R \sigma_V S V_R \cp{}{S}{V_R} C(S,V_R,V_L,t) \nn \\
&+ \sigma \rho^L \sigma_V S V_L \cp{}{S}{V_L} C(S,V_R,V_L,t) + \frac{\sigma_V^2 V_R}{2} \sop{}{V_R} C(S,V_R,V_L,t)
+ \frac{\sigma_V^2 V_L}{2} \sop{}{V_L} C(S,V_R,V_L,t) \nonumber \\
&+  \sqrt{V_R} \int_0^{\infty} \left[C(Se^y,V_R,V_L,t) - C(S,V_R,V_L,t) - \fp{}{S} C(S,V_R,V_L,t) S(e^y-1)  \right]
\lambda \frac{e^{-\nu_R |y| }}{|y|^{1+\alpha}} dy \nonumber \\
&+ \sqrt{V_L} \int^0_{-\infty} \left[C(Se^y,V_R,V_L,t) - C(S,V_R,V_L,t) - \fp{}{S} C(S,V_R,V_L,t) S(e^y-1)  \right]
\lambda \frac{e^{-\nu_L |y| }}{|y|^{1+\alpha}} dy, \nonumber
\end{align}

\ni on the domain $S>0,V_R>0,V_L>0$ and $t \in [0,T]$, where $S, V_R, V_L$ are state variables (spot price and stochastic variances). For the following we make some critical assumptions.

\begin{enumerate}
\item This PIDE could be generalized with allowance for GTSP processes, which means we substitute $\alpha$ in Eq.~(\ref{pide_init}) with $\alpha_R, \alpha_L$, and $\lambda$ with $\lambda_R, \lambda_L$ correspondingly.
\item The obtained PIDE could be solved by using a splitting technique similar to that proposed in \cite{ItkinCarrBarrierR3}.
\item We assume $\alpha_R < 0, \alpha_L < 0$ which means we consider only jumps with finite activity. Therefore, each compensator under the integral could be integrated out.
\end{enumerate}

As a result we consider just that steps of splitting which deals with the remaining integral term. The corresponding equation reads

\begin{equation} \label{int1}
\fp{}{t}C(S,V_R,V_L,t) = -\sqrt{V_R} \int_0^{\infty} C(Se^y,V_R,V_L,t) \lambda_R \dfrac{e^{-\nu_R |y|}}{|y|^{1+\alpha_R}} dy
\end{equation}

\ni for positive jumps and
\begin{equation} \label{int2}
\fp{}{t}C(S,V_R,V_L,t) = -\sqrt{V_L} \int_{-\infty}^0 C(Se^y,V_R,V_L,t) \lambda_L \dfrac{e^{-\nu_L |y|}}{|y|^{1+\alpha_L}} dy
\end{equation}

\ni for negative jumps.

Making a change of variables $x = \log S$ and omitting dependence on dummy variables $V_R, V_L$ we can rewrite these two equations in a more standard form
\begin{eqnarray} \label{intX}
\fp{}{t}C(x,t) &=& -\sqrt{V_R} \int^{\infty}_0 C(x+y,t) \lambda_R \dfrac{e^{-\nu_R |y|}}{|y|^{1+\alpha_R}} dy \\
\fp{}{t}C(x,t) &=& - \sqrt{V_L} \int_{-\infty}^0 C(x+y,t) \lambda_L \dfrac{e^{-\nu_L |y|}}{|y|^{1+\alpha_L}} dy \nn
\end{eqnarray}

To make it clear the above is not a system of equations but rather two different steps of the splitting procedure.

Now an important note is that in accordance with the definition of these integrals we can rewrite the kernel as
\begin{eqnarray} \label{intEq}
\fp{}{t}C(x,t) &=& - \sqrt{V_R} \int^{\infty}_0 C(x+y,t) \lambda_R \dfrac{e^{-\nu_R |y|}}{|y|^{1+\alpha_R}}\mathbf{1}_{y>0} dy \\
\fp{}{t}C(S,t) &=& - \sqrt{V_L} \int_{-\infty}^0 C(x+y,t) \lambda_L \dfrac{e^{-\nu_L |y|}}{|y|^{1+\alpha_L}}\mathbf{1}_{y<0} dy \nn
\end{eqnarray}

This two equations are still PIDE or evolutionary integral equations. We want to apply our new method to transform them to a certain pseudo parabolic equations.

\paragraph{First equation in the Eq.~(\ref{intEq})} Assuming $z=x+y$ we rewrite it in the form
\begin{equation} \label{intEq1}
\fp{}{t}C(x,t) = - \sqrt{V_R} \int^{\infty}_x C(z,t) \lambda_R \dfrac{e^{-\nu_R |z-x|}}{|z-x|^{1+\alpha_R}}\mathbf{1}_{z-x>0} dz
\end{equation}

To achieve our goal we have to solve the following problem. We need to find a differential operator $\mathcal{A}^+_y$ which Green's function is the kernel of the integral in the Eq.~(\ref{intEq1}), i.e.
\begin{equation}\label{greenFdef}
    \mathcal{A}^+_y\left[\lambda \dfrac{e^{-\nu |y|}}{|y|^{1+\alpha}}\mathbf{1}_{y>0}\right] = \delta(y)
\end{equation}

We prove the following proposition.
\begin{proposition} \label{p1}
Assume that in the Eq.~(\ref{greenFdef}) $\alpha \in \mathbb{I}$, and $\alpha < 0$. Then the solution of the Eq.~(\ref{greenFdef}) with respect to $\mathcal{A}^+_y$ is
\bd
\mathcal{A}^+_y = \dfrac{1}{\lambda p!}\left(\nu + \fp{}{y}\right)^{p+1} \equiv \dfrac{1}{\lambda p!}\left[
\sum_{i=0}^{p+1} C^{p+1}_{i} \nu^{p+1-i} \dfrac{\partial^i}{\partial y^i}\right], \quad p \equiv -(1 + \alpha) \ge 0,
\ed

\ni where $C^{p+1}_{i}$ are the binomial coefficients.
\end{proposition}
\begin{proof}[{\bf Proof}]
As it will be shown later this result could be proven by taking Laplace transform of both parts of the Eq.~(\ref{greenFdef}). It could be also verified using Mathematica commands given in Fig.~\ref{math} (we can check the above result for any positive integer $p$).
\begin{figure}[t!]
\begin{center}
\fbox{\includegraphics[totalheight=1.7 in]{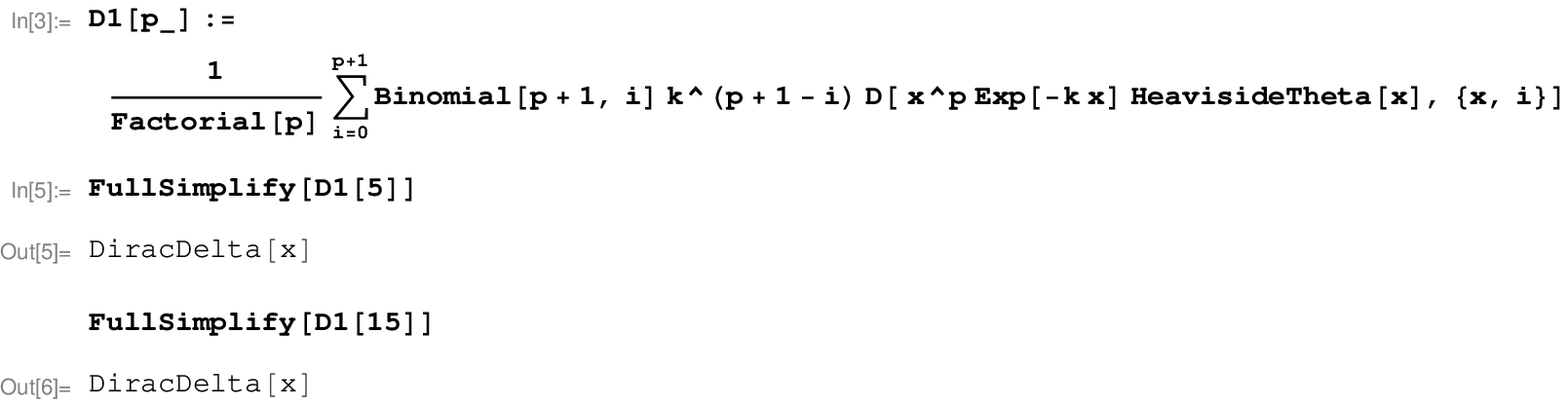}}
\caption{Mathematica commands to verify the proposition \ref{p1}.}
\label{math}
\end{center}
\end{figure}
\end{proof}

\paragraph{Second equation in the Eq.~(\ref{intEq})}
For the second equation in the Eq.~(\ref{intEq}) it is possible to elaborate an analogous approach. Again assuming $z=x+y$ we rewrite it in the form
\begin{equation} \label{intEq2}
\fp{}{t}C(x,t) = - \sqrt{V_L} \int^x_{-\infty} C(z,t) \lambda_R \dfrac{e^{-\nu_R |z-x|}}{|z-x|^{1+\alpha_R}}\mathbf{1}_{z-x<0} dz
\end{equation}

Now we need to find a differential operator $\mathcal{A}^-_y$ which Green's function is the kernel of the integral in the Eq.~(\ref{intEq2}), i.e.
\begin{equation}\label{greenFdef2}
    \mathcal{A}^-_y\left[\lambda \dfrac{e^{-\nu |y|}}{|y|^{1+\alpha}}\mathbf{1}_{y<0}\right] = \delta(y)
\end{equation}

We prove the following proposition.
\begin{proposition} \label{p2}
Assume that in the Eq.~(\ref{greenFdef2}) $\alpha \in \mathbb{I}$, and $\alpha < 0$. Then the solution of the Eq.~(\ref{greenFdef2}) with respect to $\mathcal{A}^-_y$ is
\bd
\mathcal{A}^-_y = \dfrac{1}{\lambda p!}\left(\nu - \fp{}{y}\right)^{p+1} \equiv \dfrac{1}{\lambda p!}\left[
\sum_{i=0}^{p+1} (-1)^i C^{p+1}_{i} \nu^{p+1-i} \dfrac{\partial^i}{\partial y^i}\right], \ p \equiv -(1 + \alpha),
\ed
\end{proposition}
\begin{proof}[{\bf Proof}]
Using Laplace transform or Mathematica commands given in Fig.~\ref{math1} we can check the above result for any positive integer $p$.
\begin{figure}[t!]
\begin{center}
\fbox{\includegraphics[totalheight=1.6 in]{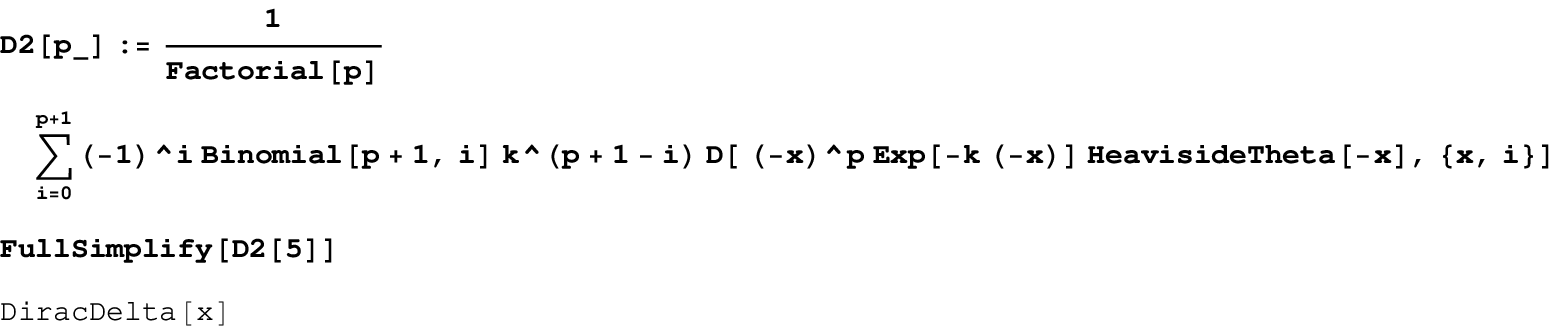}}
\caption{Mathematica commands to verify the proposition \ref{p2}.}
\label{math1}
\end{center}
\end{figure}
\end{proof}

To proceed we need to prove two other statements.
\begin{proposition} \label{p3}
Let us denote the kernels as
\begin{equation} \label{kernel}
    g^+(z-x) \equiv \lambda_R \dfrac{e^{-\nu_R |z-x|}}{|z-x|^{1+\alpha_R}}\mathbf{1}_{z-x>0}.
\end{equation}

Then
\begin{equation}\label{kernelS}
\mathcal{A}^-_x g^+(z-x) = \delta(z-x).
\end{equation}
\end{proposition}
\begin{proof}[{\bf Proof}]
\begin{align}
\mathcal{A}^-_x g^+(z-x) &= \dfrac{1}{\lambda p!}\left(\nu - \fp{}{x}\right)^{p+1} g^+(z-x) =
\dfrac{1}{\lambda p!}\left(\nu + \fp{}{(z-x)}\right)^{p+1} g^+(z-x) \nn \\
&= \mathcal{A}^+_{z-x} g^+(z-x) = \delta(z-x) \nn
\end{align}
\end{proof}

\begin{proposition} \label{p4}
Let us denote the kernels as
\begin{equation} \label{kernel1}
    g^-(z-x) \equiv \lambda_L \dfrac{e^{-\nu_L |z-x|}}{|z-x|^{1+\alpha_L}}\mathbf{1}_{z-x<0}.
\end{equation}

Then
\begin{equation}\label{kernelS1}
\mathcal{A}^+_x g^-(z-x) = \delta(z-x).
\end{equation}
\end{proposition}
\begin{proof}[{\bf Proof}]
\begin{align}
\mathcal{A}^+_x g^-(z-x) &= \dfrac{1}{\lambda p!}\left(\nu + \fp{}{x}\right)^{p+1} g^-(z-x) =
\dfrac{1}{\lambda p!}\left(\nu - \fp{}{(z-x)}\right)^{p+1} g^-(z-x) \nn \\
&= \mathcal{A}^-_{z-x} g^-(z-x) = \delta(z-x) \nn
\end{align}
\end{proof}

\paragraph{Transformation} We now apply the operator $\mathcal{A}^-_x$ to both parts of the Eq.~(\ref{intEq1})
to obtain
\begin{align} \label{operEq}
\mathcal{A}^-_x & \fp{}{t}C(x,t) = - \sqrt{V_R} \mathcal{A}^-_x \int_x^\infty C(z,t) g^+(z-x) dz
= -\sqrt{V_R} \left\{ \int_x^\infty C(z,t) \mathcal{A}^-_xg^+(z-x) dz + \mathcal{R} \right\}  \\
&=  -\sqrt{V_R} \left\{ \int_x^\infty C(z,t) \delta(z-x) dz + \mathcal{R}\right\} = -\dfrac{1}{2}\sqrt{V_R} C(x,t) - \sqrt{V_R} \mathcal{R} \nn
\end{align}

Here
\begin{equation} \label{extra}
\mathcal{R} = \sum_{i=0}^p a_i \left(\dfrac{\partial^{p-i}}{\partial x^{p-i}} V(x) \right)\left(\dfrac{\partial^i}{\partial x^i} g(z-x)\right)\Big|_{z-x=0},
\end{equation}

\ni and $a_i$ are some constant coefficients. As from the definition in the Eq.~(\ref{kernel}) $g(z-x) \propto (z-x)^p$, the only term in the Eq.~(\ref{extra}) which does not vanish is that at $i=p$. Thus
\begin{equation} \label{extra1}
\mathcal{R} = V(x) \left( \dfrac{\partial^p}{\partial x^p} g(z-x) \right) \Big|_{z-x=0} = V(x) p! \mathbf{1}_{(0)} = 0;
\end{equation}

With allowance for this expression from the  Eq.~(\ref{operEq}) we obtain the following pseudo parabolic equation for $C(x,t)$
\begin{equation}\label{operEq1}
    \mathcal{A}^-_x \fp{}{t}C(x,t) = - \dfrac{1}{2}\sqrt{V_R} C(x,t)
\end{equation}

Applying the operator $\mathcal{A}^+_x$ to both parts of the second equation in the Eq.~(\ref{intEq2}) and doing in the same way as in the previous paragraph we obtain the following pseudo parabolic equation for $C(x,t)$
\begin{equation}\label{operEq2}
    \mathcal{A}^+_x \fp{}{t}C(x,t) = -\dfrac{1}{2}\sqrt{V_L} C(x,t)
\end{equation}

\section{Solution of the pseudo parabolic equation} \label{nm1}
Assume that the inverse operator $\mathcal{A}^{-1}$ exists (see discussion later) we can represent, for instance, the Eq.~(\ref{operEq1}) in the form
\begin{equation}\label{operEq11}
   \fp{}{t}C(x,t) = - \mathcal{B} C(x,t), \quad \mathcal{B} \equiv \dfrac{1}{2}\sqrt{V_R} (\mathcal{A}^-_x)^{-1},
\end{equation}

This equation can be formally solved analytically to give
\begin{equation}\label{solOper1}
   C(x,t) = e^{\mathcal{B} (T-t)} C(x,T),
\end{equation}

\ni where $T$ is the time to maturity and $C(x,T)$ is payoff. Switching to a new variable $\tau=T-t$ to go backward in time we rewrite the Eq.~(\ref{solOper1}) as
\begin{equation}\label{solOper11}
   C(x,\tau) = e^{\mathcal{B} \tau} C(x,0),
\end{equation}

Below we consider numerical methods which allow one to compute this operator exponent with a prescribed accuracy. First we consider a straightforward approach when $\alpha \in \mathbb{I}$.

\subsection{Numerical method when $\alpha \in \mathbb{I}$} \label{nmInt}

Suppose that the whole time space is uniformly divided into $N$ steps, so the time step $\theta = T/N$ is known. Assuming that the solution at time step $k, 0 \le k < N$ is known and we go backward in time, we could rewrite the Eq.~(\ref{solOper1}) in the form
\begin{equation}\label{solOper2}
   C^{k+1}(x) = e^{\mathcal{B} \theta} C^k(x),
\end{equation}

\ni where $C^k(x) \equiv C(x,k\theta)$. To get representation of the rhs of the Eq.~(\ref{solOper2}) with given order of approximation in $\theta$, we can substitute the whole exponential operator with its Pad\'e approximation of the corresponding order $m$.

First, consider the case $m=1$. A symmetric Pad\'e approximation of the order $(1,1)$ for the exponential operator is
\begin{equation}\label{Pade1}
e^{\mathcal{B} \theta} = \dfrac{1 + \mathcal{B}\theta/2}{1 - \mathcal{B}\theta/2}
\end{equation}

Substituting this into the Eq.~(\ref{solOper2}) and affecting both parts of the equation by the operator ${1 - \mathcal{B}\theta/2}$ gives
\begin{equation}\label{discrete}
\left(1 - \dfrac{1}{2}\mathcal{B}\theta\right) C^{k+1}(x) =  \left(1 + \dfrac{1}{2}\mathcal{B}\theta\right) C^k(x).
\end{equation}

This is a discrete equation which approximates the original solution given in the Eq.~(\ref{solOper2}) with the second order in $\theta$. One can easily recognize in this scheme a famous Crank-Nicolson scheme.

We do not want to invert the operator $\mathcal{A}^-_x$ in order to compute the operator $\mathcal{B}$ because $\mathcal{B}$ is an integral operator. Therefore, we will apply the operator $\mathcal{A}^-_x$ to the both sides of the Eq.~(\ref{discrete}). The resulting equation is a pure differential equation and reads
\begin{equation}\label{discrete2}
\left( \mathcal{A}^-_x - \dfrac{\sqrt{V_R}}{4} \theta\right) C^{k+1}(x) =  \left(\mathcal{A}^-_x + \dfrac{\sqrt{V_R}}{4} \theta\right) C^k(x).
\end{equation}

Let us work with the operator $\mathcal{A}^-_x$ (for the operator $\mathcal{A}^+_x$ all corresponding results can be obtained in a similar way). The operator $\mathcal{A}^-_x$ contains derivatives in $x$ up to the order $p+1$. If one uses a finite difference representation of these derivatives the resulting matrix in the rhs of the Eq.~(\ref{discrete2}) is a band matrix. The number of diagonals in the matrix depends on the value of $p = -(1+\alpha_R) > 0$. For central difference approximation of derivatives of order $d$ in $x$ with the order of approximation $q$ the matrix will have at least $l = d+q$ diagonals, where it appears that $d+q$ is necessarily an odd number (\cite{Eberly}). Therefore, if we consider a second order approximation in $x$, i.e. $q=2$ in our case the number of diagonals is $l = p+3 = 2 - \alpha_R$.

As the rhs matrix $\mathcal{D} \equiv \mathcal{A}^-_x - \sqrt{V_R}\theta/4$ is a band matrix the solution of the corresponding system of linear equations in the Eq.~(\ref{discrete2}) could be efficiently obtained using a modern technique (for instance, using a ScaLAPACK package). The computational cost for the LU factorization of an N-by-N matrix with lower bandwidth $P$ and upper bandwidth $Q$ is $2 N P Q$ (this is an upper bound) and storage-wise - $N(P+Q)$. So in our case of the symmetric matrix the cost is
$(1-\alpha_R)^2 N/2$ performance-wise and $N(1-\alpha_R)$ storage-wise. This means that the complexity of our algorithm is still $O(N)$ while the constant $(1-\alpha_R)^2/2$ could be large.

A typical example could be if we solve our PDE using an $x$-grid with 300 nodes, so $N=300$. Suppose $\alpha_R = -10$. Then the complexity of the algorithm is $60 N = 18000$. Compare this with the FFT algorithm complexity which is $(34/9) 2N \log_2 (2N) \approx 20900$ \footnote{ We use $2N$ instead of $N$ because in order to avoid undesirable wrap-round errors a common technique is to embed a discretization Toeplitz matrix into a circulant matrix. This requires to double the initial vector of unknowns}, one can see that our algorithm is of the same speed as the FFT.

The case $m=2$ could be achieved either using symmetric (2,2) or diagonal (1,2) Pad\'e approximations of the operator exponent. The (1,2) Pad\'e approximation reads
\begin{equation}\label{Pade21}
e^{\mathcal{B} \theta} = \dfrac{1 + \mathcal{B}\theta/3}{1 - 2\mathcal{B}\theta/3 + \mathcal{B}^2\theta^2/6},
\end{equation}

\ni and the corresponding finite difference scheme for the solution of the Eq.~(\ref{solOper2}) is
\begin{equation}\label{fd21}
\left[(\mathcal{A}^-_x)^2 - \dfrac{1}{3}\sqrt{V_R} \theta \mathcal{A}^-_x + \dfrac{1}{24}V_R \theta^2 \right]
C^{k+1}(x) =  \mathcal{A}^-_x\left[\mathcal{A}^-_x + \dfrac{1}{6}\sqrt{V_R} \theta\right] C^k(x).
\end{equation}

\ni which is of the third order in $\theta$. The (2,2) Pad\'e approximation is
\begin{equation}\label{Pade22}
e^{\mathcal{B} \theta} = \dfrac{1 + \mathcal{B}\theta/2 + \mathcal{B}^2\theta^2/12}{1 - \mathcal{B}\theta/2 + \mathcal{B}^2\theta^2/12},
\end{equation}

and the corresponding finite difference scheme for the solution of the Eq.~(\ref{solOper2}) is
\begin{equation}\label{fd22}
\left[(\mathcal{A}^-_x)^2 - \dfrac{1}{4}\sqrt{V_R} \theta\mathcal{A}^-_x + \dfrac{1}{48}V_R \theta^2 \right]
C^{k+1}(x) =  \left[(\mathcal{A}^-_x)^2 + \dfrac{1}{4}\sqrt{V_R} \theta \mathcal{A}^-_x + \dfrac{1}{48}V_R \theta^2 \right] C^k(x),
\end{equation}

\ni which is of the fourth order in $\theta$.

Matrix of the operator $(\mathcal{A}^-_x)^2$ has $2l-1$ diagonals, where $l$ is the number of diagonals of the matrix $\mathcal{A}^-_x$. Thus, the finite difference equations Eq.~(\ref{fd21}) and Eq.~(\ref{fd22}) still have band matrices and could be efficiently solved using an appropriate technique.

\subsection{Stability analysis}

Stability analysis of the derived finite difference schemes could be provided using a standard von-Neumann method. Suppose that operator $\mathcal{A}^-_x$ has eigenvalues $\zeta$ which belong to continuous spectrum. Any finite difference approximation of the operator $\mathcal{A}^-_x$ - $FD(\mathcal{A}^-_x)$ - transforms this continuous spectrum into some discrete spectrum, so we denote the eigenvalues of the discrete operator $FD(\mathcal{A}^-_x)$ as $\zeta_i, i=1,N$, where $N$ is the total size of the finite difference grid.

Now let us consider, for example, the Crank-Nicolson scheme given in the Eq.~(\ref{discrete2}). It is stable if in some norm $\|\cdot\|$
\begin{equation}\label{norm}
    \Bigg\| \left( \mathcal{A}^-_x - \dfrac{\sqrt{V_R}}{4} \theta\right)^{-1}\left(\mathcal{A}^-_x + \dfrac{\sqrt{V_R}}{4} \theta\right) \Bigg\| < 1.
\end{equation}

It is easy to see that this inequality obeys when all eigenvalues of the operator $\mathcal{A}^-_x$ are negative. However, based on the definition of this operator given in the Proposition~\ref{p2}, it is clear that the central finite difference approximation of the first derivative does not give rise to a full negative spectrum of eigenvalues of the operator $FD(\mathcal{A}^-_x)$. So below we define a different approximation.

\paragraph{Case ${\bf \alpha_R < 0}$.}
Therefore, in this case we will use a one-sided forward approximation of the first derivative which is a part of the operator $\left(\nu_R - \dfrac{\partial}{\partial x}\right)^{\alpha_R}$. Define $h = (x_{max} - x_{min})/N$ to be the grid step in the $x$-direction, $N$ is the total number of steps, $x_{min}$ and $x_{max}$ are the left and right boundaries of the grid. Also define $c_i^k = C^k(x_i)$.
To make our method to be of the second order in $x$ we use the following numerical approximation
\begin{equation}\label{forward1der}
    \fp{C^k(x)}{x} = \dfrac{-C_{i+2}^k + 4 C_{i+1}^k - 3 C_{i}^k}{2 h} + O(h^2)
\end{equation}

Matrix of this discrete difference operator has the following form
\begin{equation}\label{mat2}
M_f = \dfrac{1}{2h}\left(
\begin{array}{ccccc}
-3 & 4 & -1 & 0 & ...0 \\
0 & -3 & 4 & -1 & ...0 \\
0 & 0 & -3 & 4 & ...0 \\
.. & .. & .. & .. & .. \\
0 & 0... & 0 & 0 & -3 \\
\end{array}
\right)
\end{equation}

All eigenvalues of $M_f$ are equal to $-3/(2 h)$.

To get a power of the matrix $M$ we use its spectral decomposition, i.e. we represent it in the form $M = E D E'$, where $D$ is a diagonal matrix of eigenvalues $d_{i}, i=1,N$ of the matrix $M$, and $E$ is a matrix of eigenvectors of the matrix $M$. Then $M^{p+1} = E D^{p+1} E'$, where the matrix $D^{p+1}$ is a diagonal matrix with elements $d_{i}^{p+1}, i=1,N$. Therefore, the eigenvalues of the matrix $\left(\nu_R - \fp{}{x}\right)^{\alpha_R}$ are $\left[\nu_R + 3/(2h)\right]^{\alpha_R}$. And, consequently, the eigenvalues of the matrix $\mathbb{B}$ are
\begin{equation}\label{eigB}
    \zeta_\mathbb{B} = \sqrt{V_R} \lambda_R \Gamma(-\alpha_R) \left\{\left[\nu_R + 3/(2h)\right]^{\alpha_R} - \nu_R^{\alpha_R} \right\}.
\end{equation}

As $\alpha_R < 0$ and $\nu_R > 0$ it follows that $\zeta_\mathbb{B} < 0$.
Rewriting the Eq.~(\ref{discrete}) in the form
\begin{equation}\label{CNgen1}
C^{k+1}(x) =  \left(1 - \dfrac{1}{2}\mathcal{B}\theta\right)^{-1} \left(1 + \dfrac{1}{2}\mathcal{B}\theta\right) C^k(x),
\end{equation}

\ni and taking into account that $\zeta_\mathbb{B} < 0$ we arrive at the following result
\begin{equation}\label{norm1}
    \Bigg\| \left(1 - \dfrac{1}{2}\mathcal{B}\theta\right)^{-1} \left(1 + \dfrac{1}{2}\mathcal{B}\theta\right) \Bigg\| < 1.
\end{equation}

We also obey the condition $\mathbb{R}\left(\nu_R - \fp{}{x}\right) > 0$. Thus, our numerical method is unconditionally stable.

\paragraph{Case ${\bf \alpha_L < 0}$.} In this case we will use a one-sided backward approximation of the first derivative in the operator $\left(\nu_L + \dfrac{\partial}{\partial x}\right)^{\alpha_L}$ which reads
\begin{equation}\label{backward1der}
    \fp{C^k(x)}{x} = \dfrac{3 C_{i}^k - 4 C_{i-1}^k + C_{i-2}^k}{2 h} + O(h^2)
\end{equation}

Matrix of this discrete difference operator has the following form
\begin{equation}\label{mat1}
M_b = \dfrac{1}{2h}\left(
\begin{array}{ccccc}
3 & 0 & 0 & 0 & ...0 \\
-4 & 3 & 0 & 0 & ...0 \\
1 & -4 & 3 & 0 & ...0 \\
.. & .. & .. & .. & .. \\
0 & 0... & 1 & -4 & 3 \\
\end{array}
\right)
\end{equation}

All eigenvalues of $M_b$ are equal to $3/(2 h)$. Then doing in a similar way as above we can show that the eigenvalues of the operator $\mathbb{B}$ read
\begin{equation}\label{eigB2}
    \zeta_\mathbb{B} = \sqrt{V_L} \lambda_L \Gamma(-\alpha_L) \left\{\left[\nu_L + 3/(2h)\right]^{\alpha_L} - \nu_L^{\alpha_L} \right\}.
\end{equation}

As $\alpha_L < 0$ and $\nu_L > 0$ it follows that $\zeta_\mathbb{B} < 0, \ \mathbb{R}\left(\nu_L + \fp{}{x}\right) > 0$, and the numerical method in this case is unconditionally stable.

\subsection{Numerical examples}

Here we describe two series of numerical experiments. In the first series we solve the equation
\begin{equation} \label{numEq}
\fp{}{\tau} C(x,\tau) = \int_0^{\infty} C(x+y,\tau)\lambda_R \dfrac{e^{-\nu_R |y|}}{|y|^{1+\alpha_R}} dy,
\quad \alpha_R < -1
\end{equation}

by using FFT and the finite difference scheme constructed based on computation of the Eq.~(\ref{numEq}) with $\alpha_R \in \mathbb{I}$ and interpolation as it was described in section \ref{nmInt}.

We solve an initial problem despite it is easy to consider a boundary problem as well. We consider a put option with time to maturity $T$ = 30 days. As the terminal condition (we compute the solution backward in time) we chose a Black-Scholes put value at $\tau=0$ where the interest rate is $r = 0.01$, the volatility is $0.1$ and the strike is $K=100$. We create a uniform grid in time with $N_t = 50$ nodes, so $\theta = T/N_t$ is the step in time.

\paragraph{FFT.}
To apply an FFT approach we first select a domain in $x$ space where the values of function $C(x,\tau$) are of our interest. Suppose this is $x \in (-x_*, x_*)$. We define a uniform grid in this domain which contains $N$ points: $x_1 = -x_*, x_2,...x_{N-1}, x_N = x_*$ such that $x_i - x_{i-1} = h, i=2...N$. We then approximate the integral in the rhs of the Eq.~(\ref{numEq}) with the first order of accuracy in $h$ as
\begin{equation} \label{numEq1}
\int_0^{\infty} C(x+y,\tau)\lambda_R \dfrac{e^{-\nu_R |y|}}{|y|^{1+\alpha_R}} dy =
\dfrac{h}{2} \sum_{j=1-i}^{N-i} C_{i+j}(\tau) f_j, \quad f_j \equiv \lambda_R \dfrac{e^{-\nu_R |x_j|}}{|x_j|^{1+\alpha_R}} + O(h^2).
\end{equation}

This approximation means that we have to extend our computational domain to the left up to $x_{1-N} = x_1 - h N$.

The matrix $|f|$ is a Toeplitz matrix. Using FFT directly to compute a matrix-vector product in the Eq.~(\ref{numEq1}) will produce a wrap-round error that significantly lowers the accuracy. Therefore a standard technique is to embed this Toeplitz matrix into a circulant matrix $\mathcal{F}$ which is defined as follows. The first row of $F$ is
\bd
F_1 = (f_0, f_1, ..., f_{N-1},0, f_{1-N},...,f_{-1}),
\ed
\ni and others are generated by permutation (see, for instance, \cite{ZhangWang2009}). We also define a vector
\bd
\hat{C} = [C_1(\tau),...C_N(\tau),\underbrace{0,...,0}_N]^T.
\ed
Then the matrix-vector product in the rhs Eq.~(\ref{numEq1}) is given by the first N rows in the vector $V = \mbox{ifft}(\mbox{fft}(F_1)*\mbox{fft}(\hat{C}))$, where \mbox{fft} and \mbox{ifft} are the forward and inverse discrete Fourier transforms as they are defined, say in Matlab. In practice, an error at edge points close to $x_1$ and $x_N$ is higher, therefore it is useful first to add some points left to $x_1$ and right to $x_N$ and then apply the above described algorithm to compute the integral. We investigated some test problems, for instance, where the function $C$ was chosen as $C(x) = x$ so the integral can be computed analytically. Based on the obtained results we found that it is useful to extend the computational domain adding $N/2$ points left to $x_1$ and right to $x_N$ that provides an accurate solution in the domain $x_1,...,x_N$. The drawback of this is that the resulting circulant matrix has $4N$ x $4N$ elements that increases the computational work by 4 times ($4N \log_2(4N) \approx 4(N \log_2 N)$).

In our calculations we used $x_* = 20, h = 2x_*/N$ regardless of the value of $N$ which varies in the experiments. Then we extended the domain to $x_1 = -x_* - h(N/2-1), x_N = x_* + h(N/2+1)$, and so this doubles the originally chosen value of $N$, i.e. $N_{new} = 2N$. But the final results were analyzed at the domain $x \in (-x_*,x*)$.

Integrating the Eq.~(\ref{numEq}) in time we use an explicit Euler scheme of the first order which is pretty fast. This is done in order to provide the worst case scenario for the below FD scheme. Thus, if our FD scheme is comparable in speed with FFT in this situation it will even better if some other more accurate integration schemes are applied together with the FFT.

\paragraph{FD.}
We build a fixed grid in the $x$ space by choosing $S_{min}  = 10^{-8}, S_{max} = 500, x_{1} = \log(S_{min}),
x_{N} = \log(S_{max}), h = (x_N - x_1)/N, N = 256$. A one-sided forward approximation of the first derivative was used as it is defined in the Eq.~(\ref{forward1der}) to approximate the operator in the Eq.~(\ref{numEq}). In the particular case considered here in our experiments $V_R \equiv 1$, and the compensators in the Eq.~(\ref{pide_init}) are not considered, because they could be integrated out at $\alpha_R < 0$ and added to the diffusion terms. The Crank-Nicolson scheme Eq.~(\ref{CNgen1}) was applied to integrate the Eq.~(\ref{numEq}) in time.

\paragraph{Results}
The first series of tests was provided when $\alpha_R \in \mathbb{I}$ and $\nu_R = 1, \lambda_R = 0.2$. The results of this series are presented in Fig.~\ref{Fig1}-\ref{Fig6}.
\begin{figure}[t!]
\vspace{0.1 in}
\hspace{0.5 in}
\centering
\fbox{\includegraphics[totalheight=3.5in]{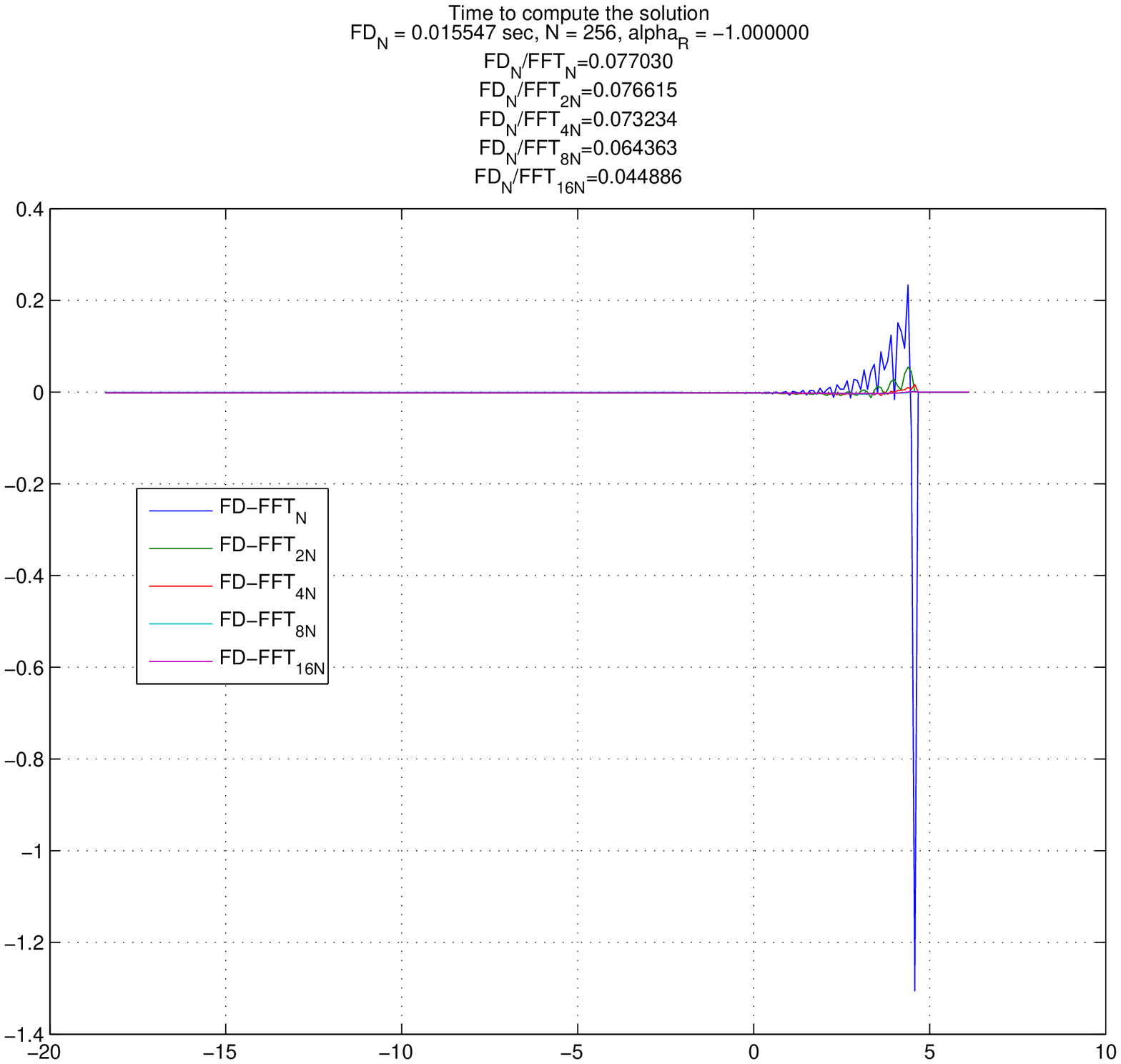}}
\caption{Difference (FD-FFT) in solutions of the Eq.~(\ref{numEq}) as a function of $x$ obtained using our finite-difference method (FD) and an explicit Euler scheme in time where the jump integral is computed using FFT. $\alpha_R = -1$.}
\label{Fig1}
\end{figure}

\begin{figure}[t!]
\vspace{0.1 in}
\hspace{0.5 in}
\centering
\fbox{\includegraphics[totalheight=3.5in]{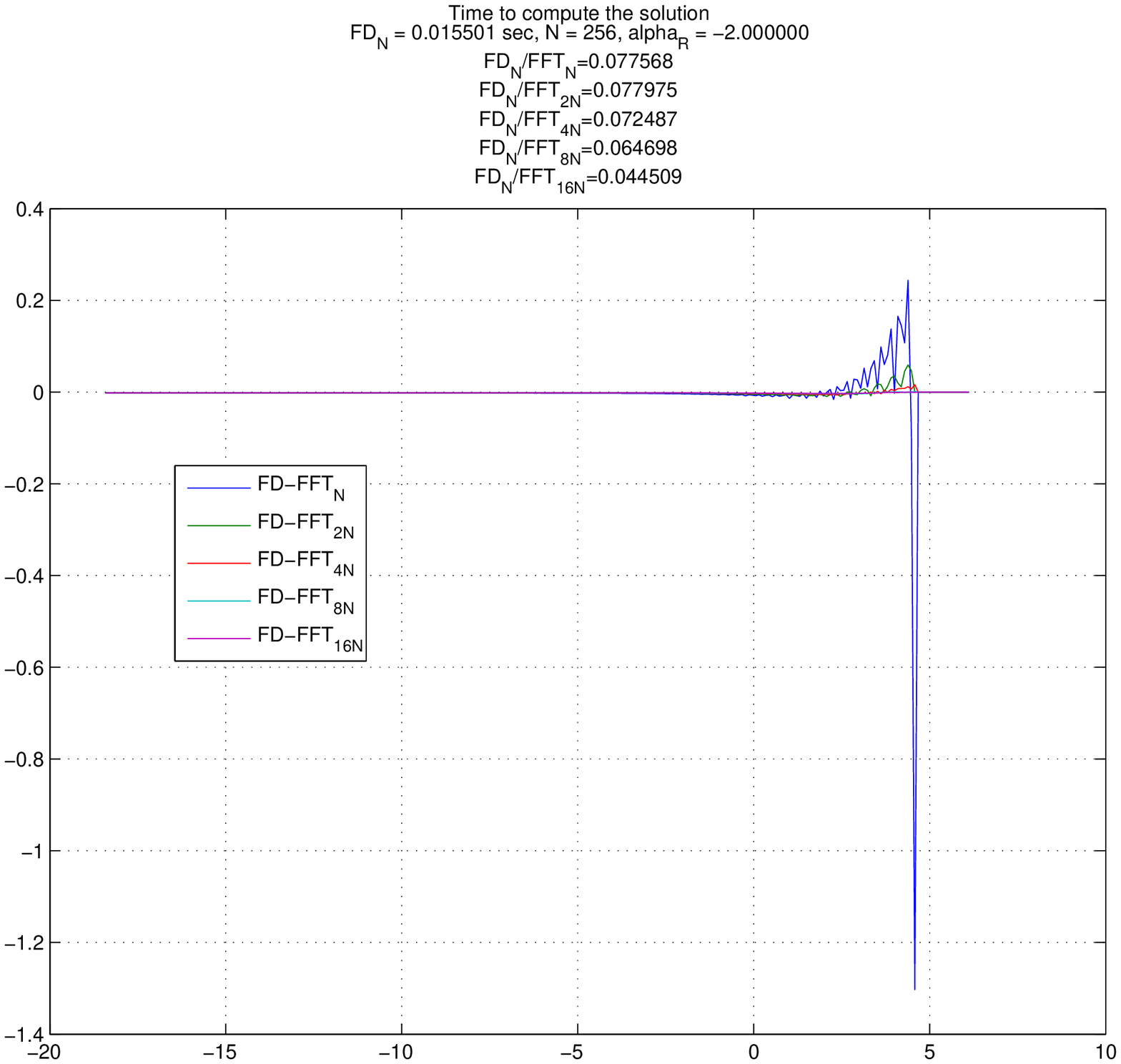}}
\caption{Same as in Fig.~\ref{Fig1}. $\alpha_R = -2$.}
\label{Fig2}
\end{figure}

\begin{figure}[t!]
\vspace{0.1 in}
\hspace{0.5 in}
\centering
\fbox{\includegraphics[totalheight=3.5in]{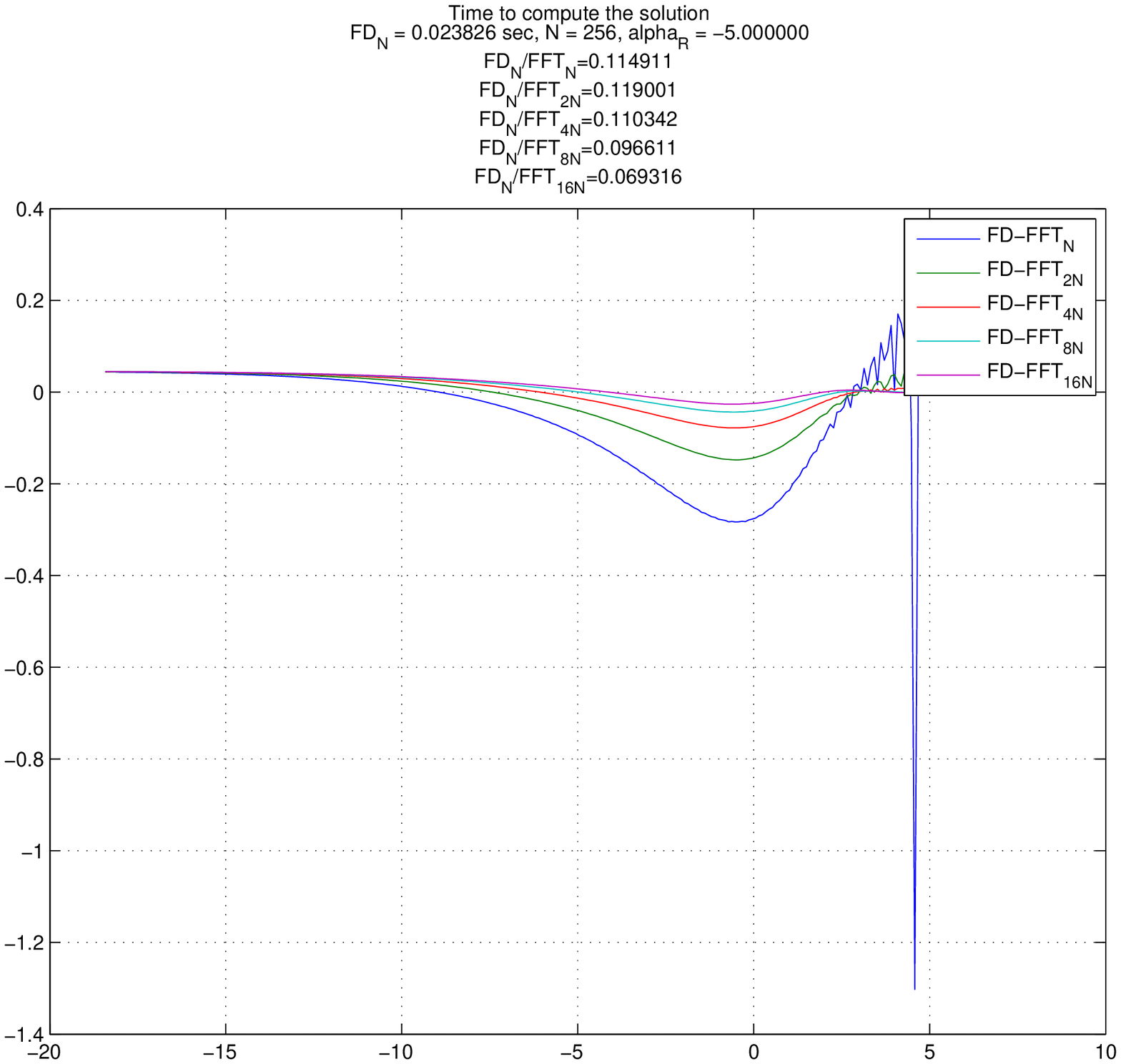}}
\caption{Same as in Fig.~\ref{Fig1}. $\alpha_R = -5$.}\label{Fig5}
\end{figure}

\begin{figure}[t!]
\vspace{0.1 in}
\hspace{0.5 in}
\centering
\fbox{\includegraphics[totalheight=3.5in]{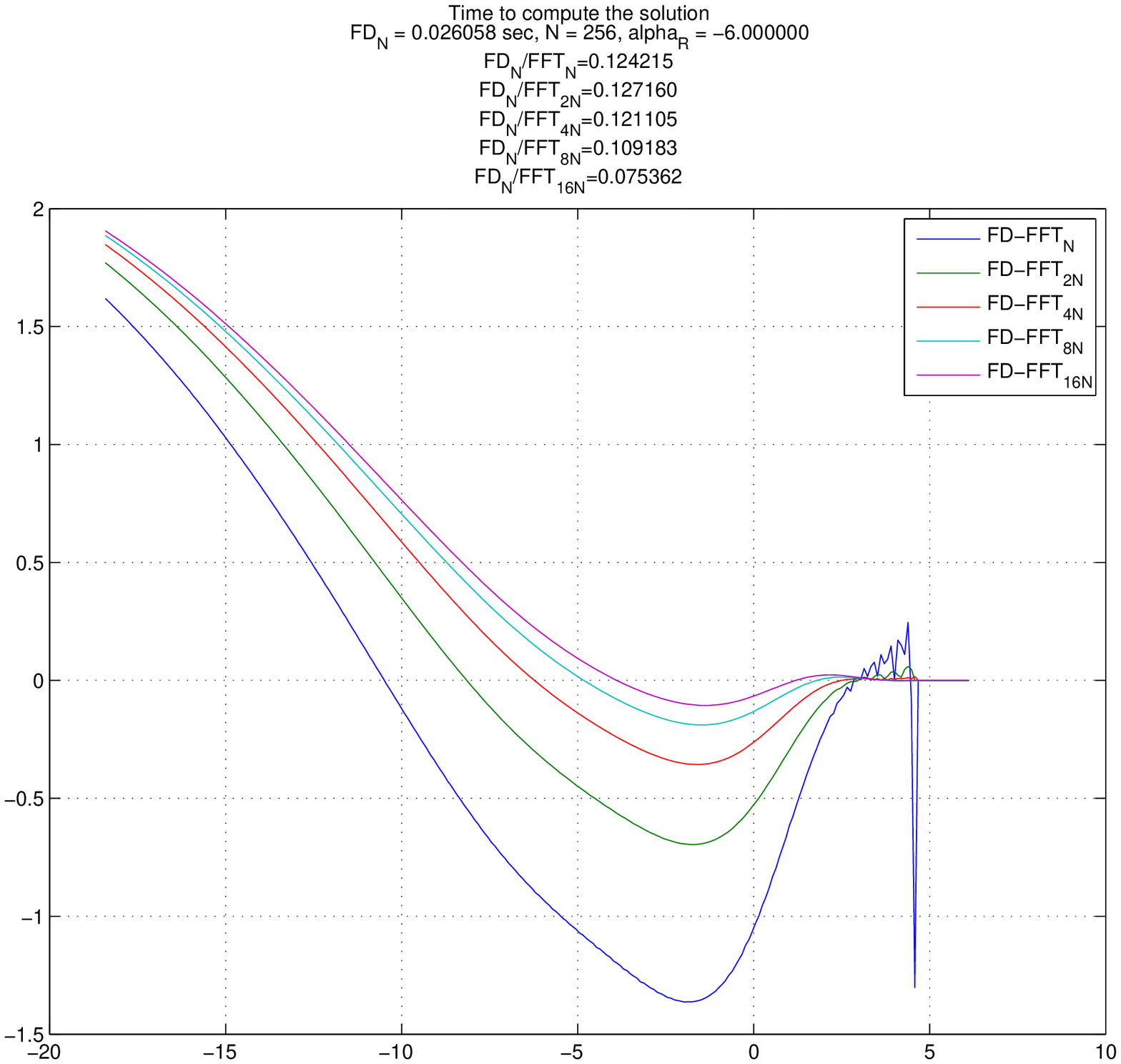}}
\caption{Same as in Fig.~\ref{Fig1}. $\alpha_R = -6$.}
\label{Fig6}
\end{figure}

In case $\alpha_R = -1$ in Fig.~\ref{Fig1} the FFT solution computed with $N=256$ provides a relatively big error which disappears with $N$ increasing. It is clear, because the Crank-Nicolson scheme is of the second order in $h$ while the approximation Eq.~(\ref{numEq1}) of the integral is of the first order in $h$. Numerical values of the corresponding steps in the described experiments are given in Tab.~\ref{Tabst}.

\begin{table}[t!]
\begin{center}
\begin{tabular}{|c|c|c|c|c|c|c|}
\hline
& FD$_{256}$ & FFT$_{256}$ & FFT$_{512}$ & FFT$_{1024}$ & FFT$_{2048}$ & FFT$_{4096}$ \cr
\hline
h & 0.096 & 0.1563 & 0.078 & 0.039 & 0.0195 & 0.00977 \cr
\hline
\end{tabular}
\caption{Grid steps $h$ used in the numerical experiments}
\label{Tabst}
\end{center}
\end{table}

Therefore, $h_{FD}^2 \approx h_{FFT_{16}}$. Actually, the difference between the FD solution with $N_{FD}=256$ and the FFT one with $N=4N_{FD}$ is almost negligible. However, the FD solution is computed almost 13 times faster. Even the FFT solution with $N=N_{FD}$ is 10 times slower than the FD one\footnote{It actually uses $4N$ points as it was already discussed}.

For $\alpha_R = -2$ in Fig.~\ref{Fig2} we see almost the same picture. For $\alpha_R = -5$ speed characteristics of both solutions are almost same while the accuracy of the FD solution decreases. This is especially pronounced for $\alpha_R = -6$ in Fig.~\ref{Fig6} at low values of $x$. The problem is that when $\alpha_R$  decreases the eigenvalues of matrix $\mathcal{B}$ in the Eq.~(\ref{discrete}) grow significantly (in our tests at $\alpha_R = -6$ the eigenvalues are of order of $10^{7}$), so in the Eq.~(\ref{norm1}) the norm of matrix is very close to 1. Thus the FD method becomes just an A-stable. However, a significant difference is observed mostly at very low values of $x$ which correspond to the spot price $S = \exp(x)$ close to zero. For a boundary problem this effect is partly dumped by the boundary condition at the low end of the domain.

The second series of tests deals with $\alpha_R \in \mathbb{R}$ using the same parameters $\nu_R = 1, \lambda_R = 0.2$. The results of this series are presented in Fig.~\ref{Fig1_5}-\ref{Fig0_5}. Four point cubic interpolation is used to compute the value of $C(x,\tau)$ at real $\alpha_R$ using the closest four integer values of $\alpha_R$.

\begin{figure}[t!]
\vspace{0.1 in}
\hspace{0.5 in}
\centering
\fbox{\includegraphics[totalheight=3.5in]{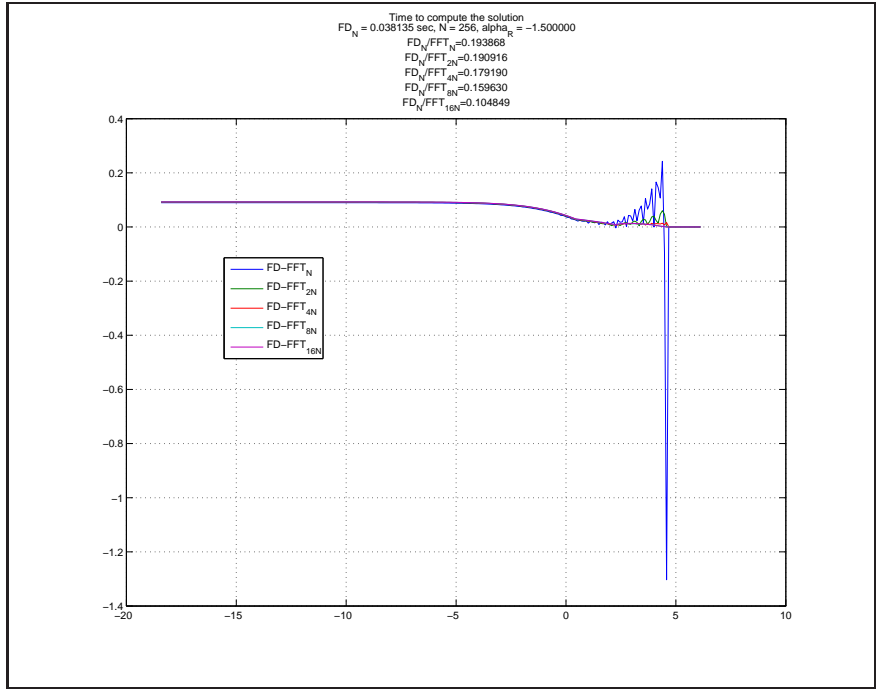}}
\caption{Difference (FD-FFT) in solutions of the Eq.~(\ref{numEq}) as a function of $x$ at $\alpha_R \in \mathbb{R}$ obtained using our finite-difference method (FD) and interpolation and an explicit Euler scheme in time where the jump integral is computed using FFT. $\alpha_R = -1.5$.}
\label{Fig1_5}
\end{figure}

\begin{figure}[t!]
\vspace{0.1 in}
\hspace{0.5 in}
\centering
\fbox{\includegraphics[totalheight=3.5in]{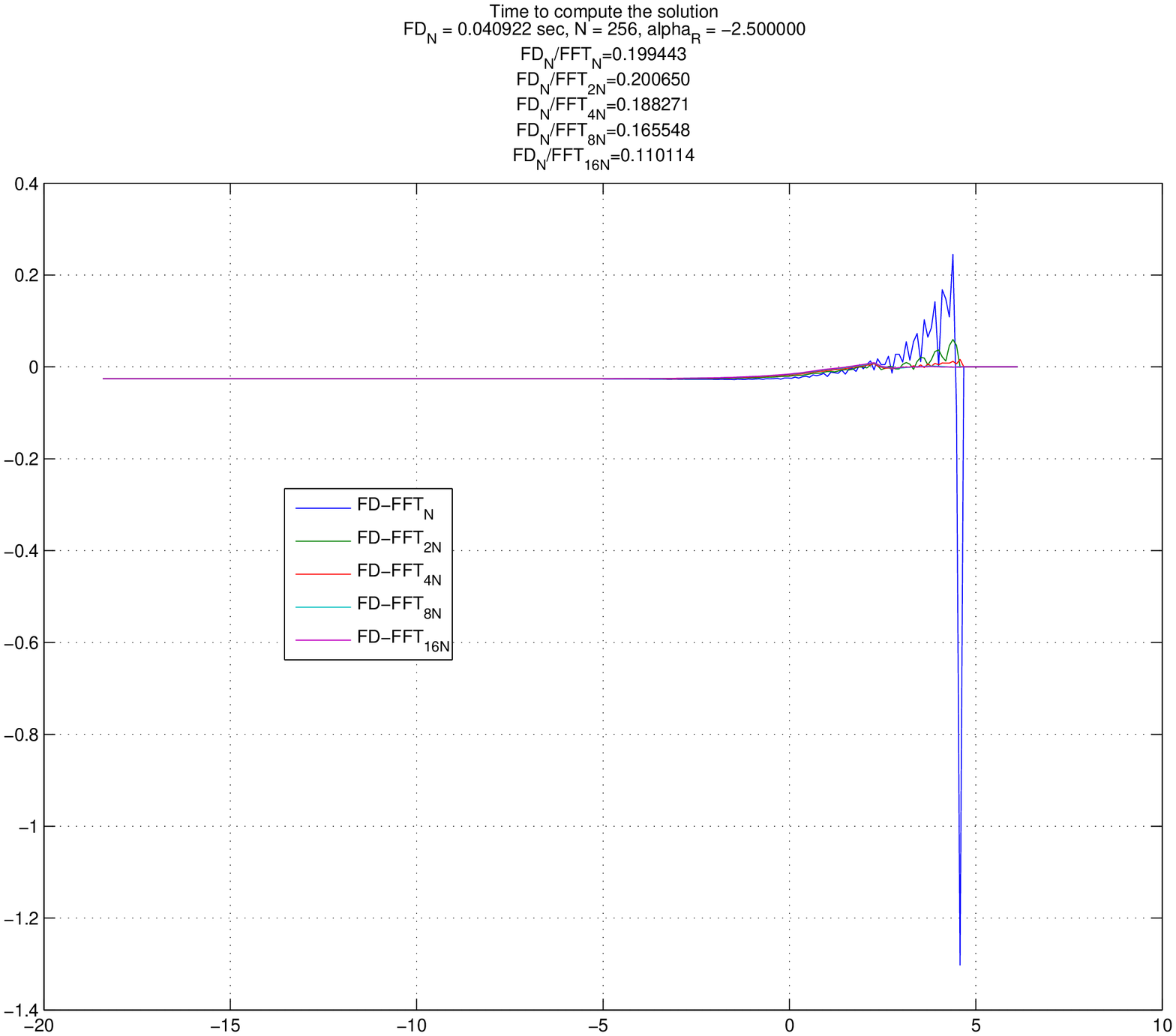}}
\caption{Same as in Fig.~\ref{Fig1_5}. $\alpha_R = -2.5$.}
\label{Fig2_5}
\end{figure}

\begin{figure}[t!]
\vspace{0.1 in}
\hspace{0.5 in}
\centering
\fbox{\includegraphics[totalheight=3.5in]{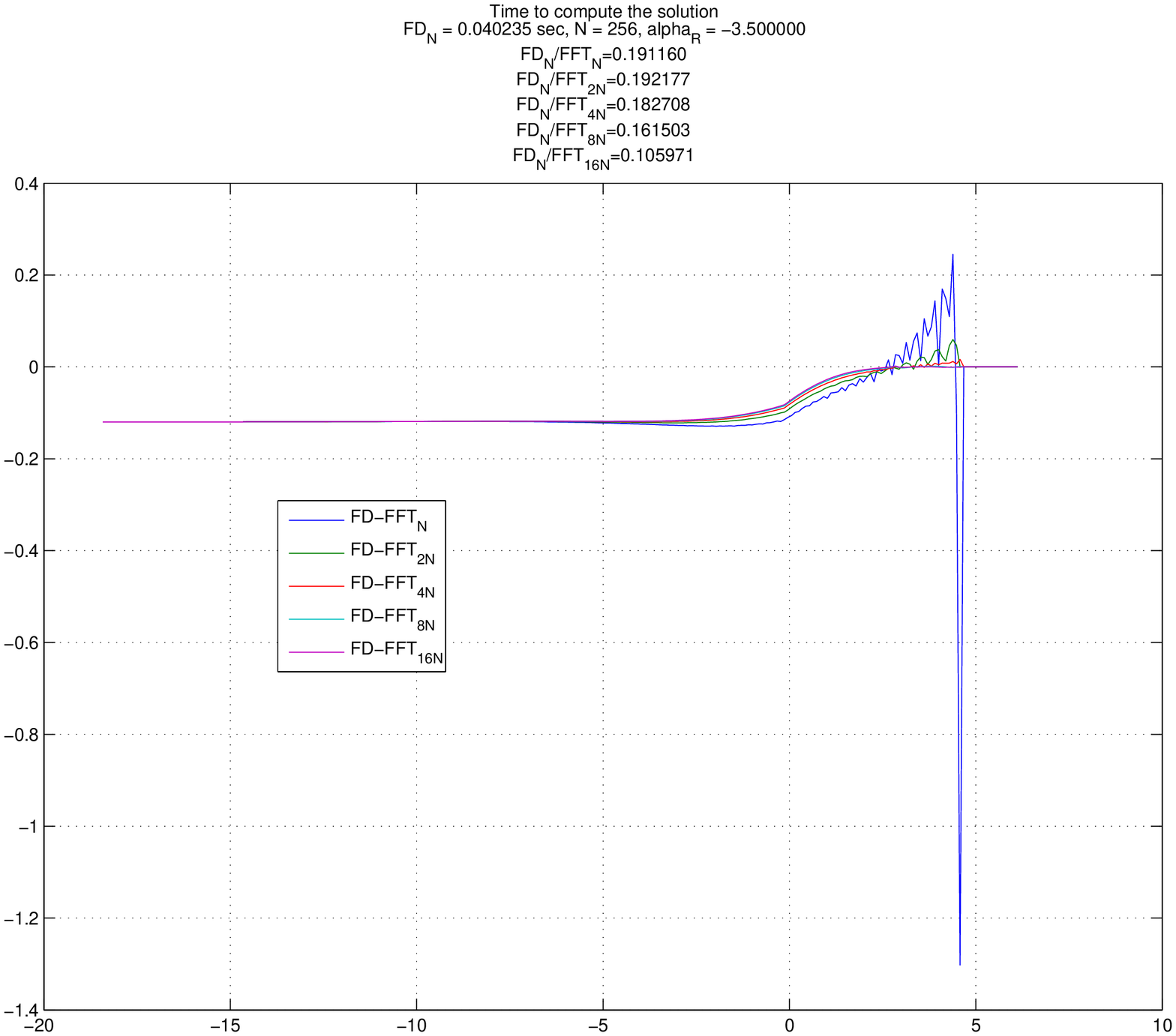}}
\caption{Same as in Fig.~\ref{Fig1_5}. $\alpha_R = -3.5$.}
\label{Fig3_5}
\end{figure}

\begin{figure}[t!]
\vspace{0.1 in}
\hspace{0.5 in}
\centering
\fbox{\includegraphics[totalheight=3.5in]{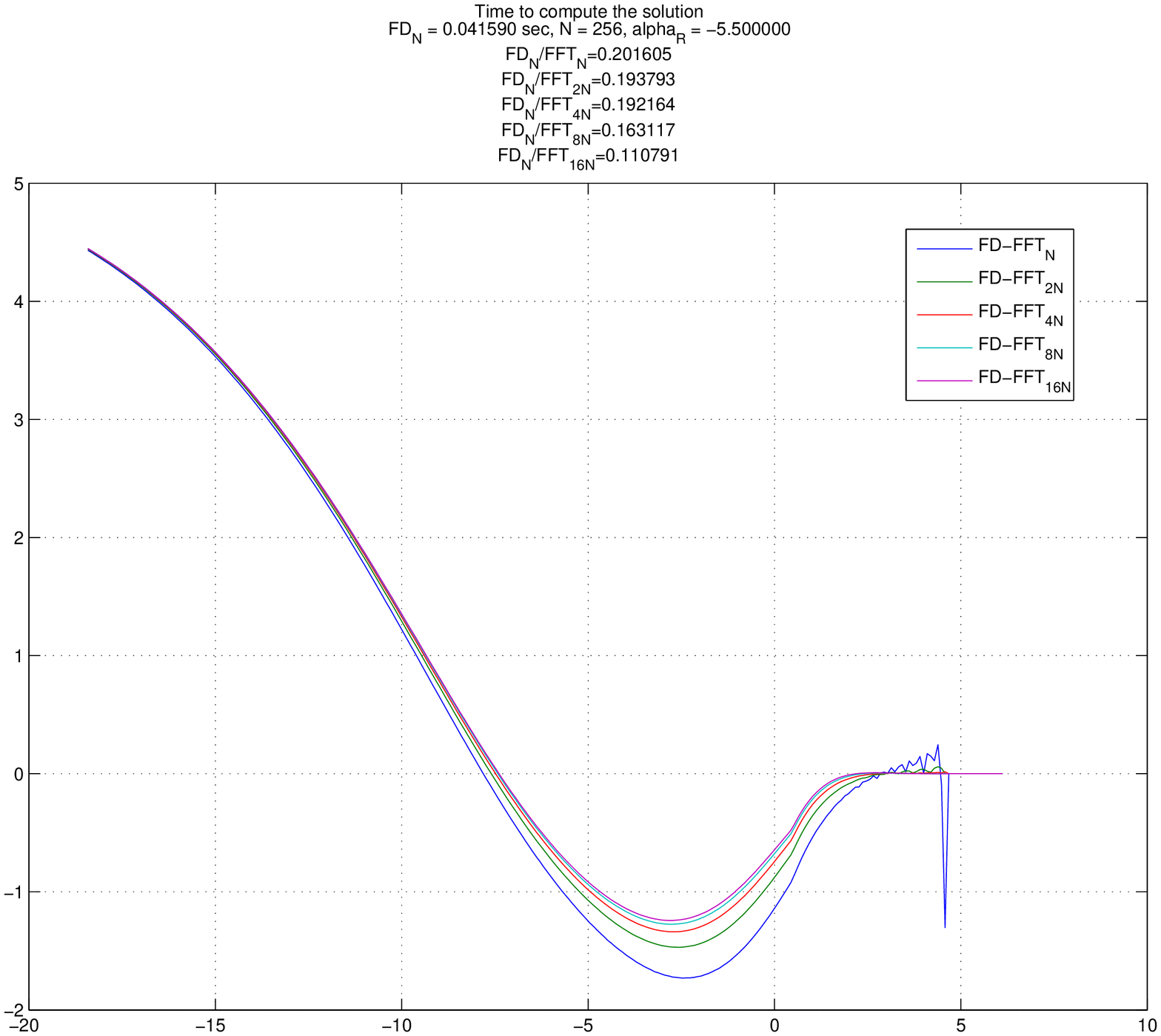}}
\caption{Same as in Fig.~\ref{Fig1_5}. $\alpha_R = -5.5$.}
\label{Fig5_5}
\end{figure}

\begin{figure}[t!]
\vspace{0.1 in}
\hspace{0.5 in}
\centering
\fbox{\includegraphics[totalheight=3.5in]{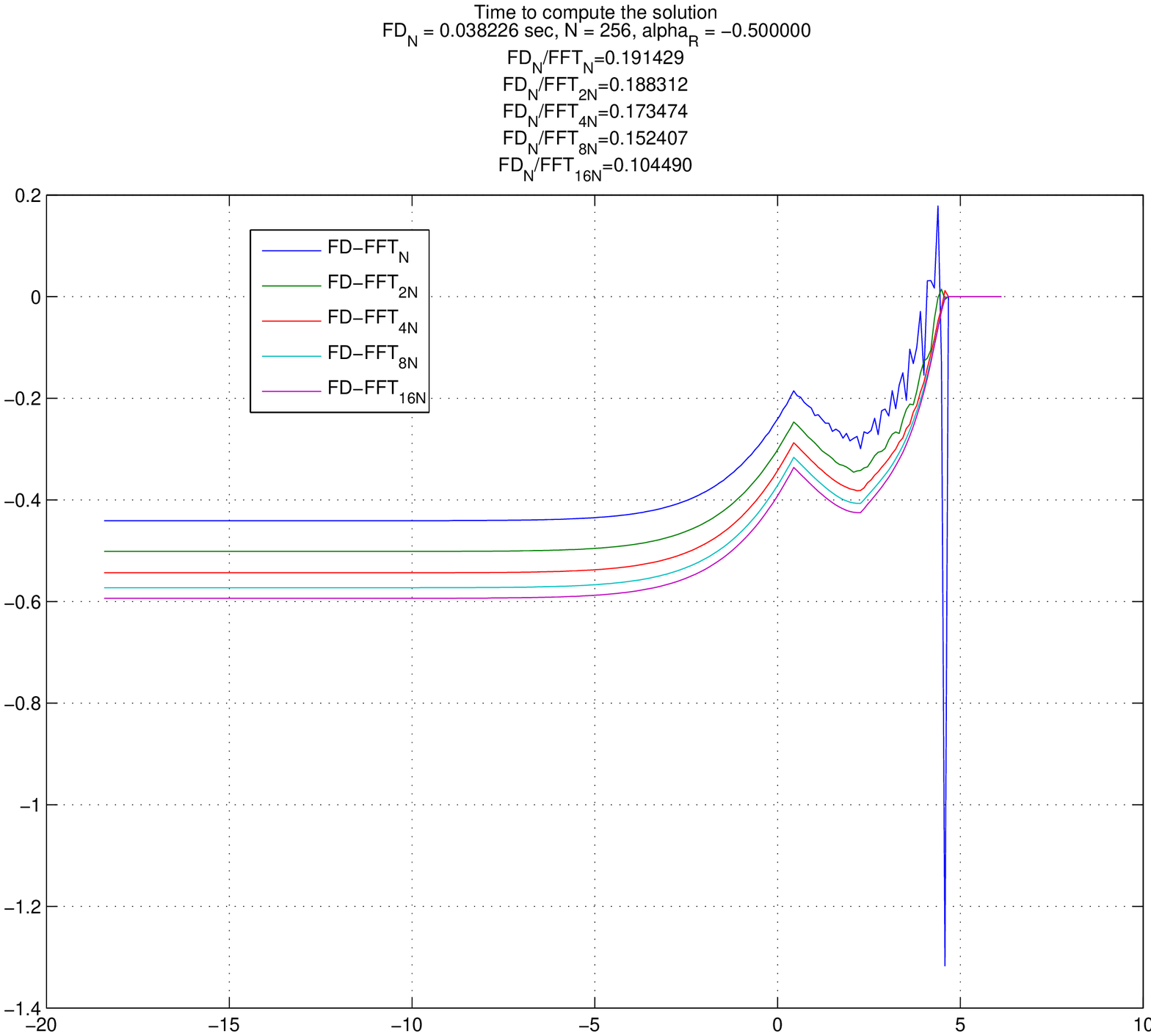}}
\caption{Same as in Fig.~\ref{Fig1_5}. $\alpha_R = -0.5$.}
\label{Fig0_5}
\end{figure}

It is seen that cubic interpolation provides pretty good approximation to the solution which is comparable with the FFT method in the accuracy and is faster in speed. Again, as we already discussed at $\alpha_R < 5$ the accuracy of the FD scheme drops down even for $\alpha_R \in \mathbb{R}$, therefore the same picture is observed for $\alpha_R \in \mathbb{R}$.

At $-1 < \alpha_R < 0$ (see Fig.~\ref{Fig0_5}) the difference between FD and FFT solutions surprisingly increases with $N$, used in the FFT method, increasing. To better understand what is the reason of that we fulfilled a test calculation of the integral in the rhs of the Eq.~(\ref{numEq}) when $C(x,\tau)$ is a known function, namely $C(x,\tau) \equiv x$. In this case this integral can be computed analytically which gives
\begin{equation} \label{testEq}
\int^\infty_0 (x+y)\dfrac{e^{-\nu_R|y|}}{|y|^{1+\alpha_R}} dy = (x\nu_R - \alpha_R) \nu_R^{\alpha_R-1}\Gamma(-\alpha_R).
\end{equation}

Then we apply the above described FFT approach and compare the numerical solution with the analytical one. The results of this test are given in Fig.~\ref{test_problem}. It is seen that FFT algorithm used in our calculations doesn't provide a good approximation to the analytical solutions at low $N$. So we expect this behavior of the FFT method occurred in our numerical experiments at $\alpha_R = -0.5$, but this doesn't explain the observed effect.

A plausible explanation is that at $\alpha_R$ close to $0$ the integral kernel becomes singular. That is why in \cite{ContVolchkova2003}
the part of the infinitesimal generator corresponding to small jumps is approximated by a differential operator of second order (additional diffusion component).
As we didn't use this technique here, an increase of $N$ forces the distance between $y=0$ and the closest FFT node boundary to become smaller, thus the kernel becomes larger.

The other reason for the FD solution to differ from the FFT solution is that at $-1 < \alpha_R < 0$ we don't use the option values computed at $\alpha_R = 0$ (remember, this is a special case that was discussed earlier). Thus, instead of interpolation we use extrapolation that certainly decreases the accuracy of the FD solution. We will resolve this problem in the next section.

\begin{figure}[t!]
\vspace{0.1 in}
\hspace{0.5 in}
\centering
\fbox{\includegraphics[totalheight=3.5in]{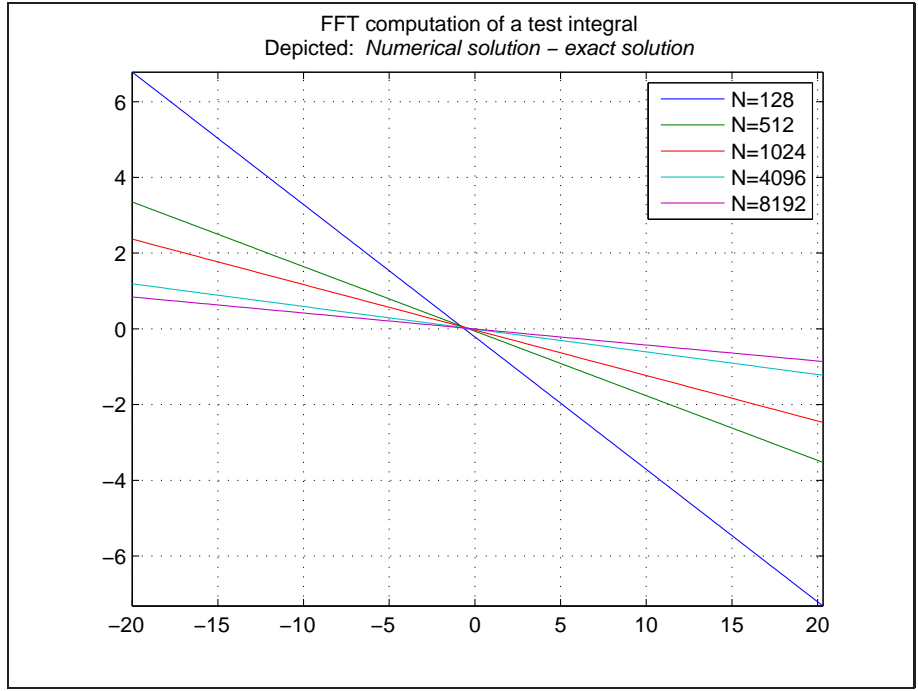}}
\caption{FFT computation of a test integral in the Eq.~(\ref{testEq})}
\label{test_problem}
\end{figure}

At the end of this section we present the option values computed using such a scheme as a function of $x$ obtained in the same test (Fig.~\ref{Calphas}).
\begin{figure}[t!]
\vspace{0.1 in}
\hspace{0.5 in}
\centering
\fbox{\includegraphics[totalheight=3.5in]{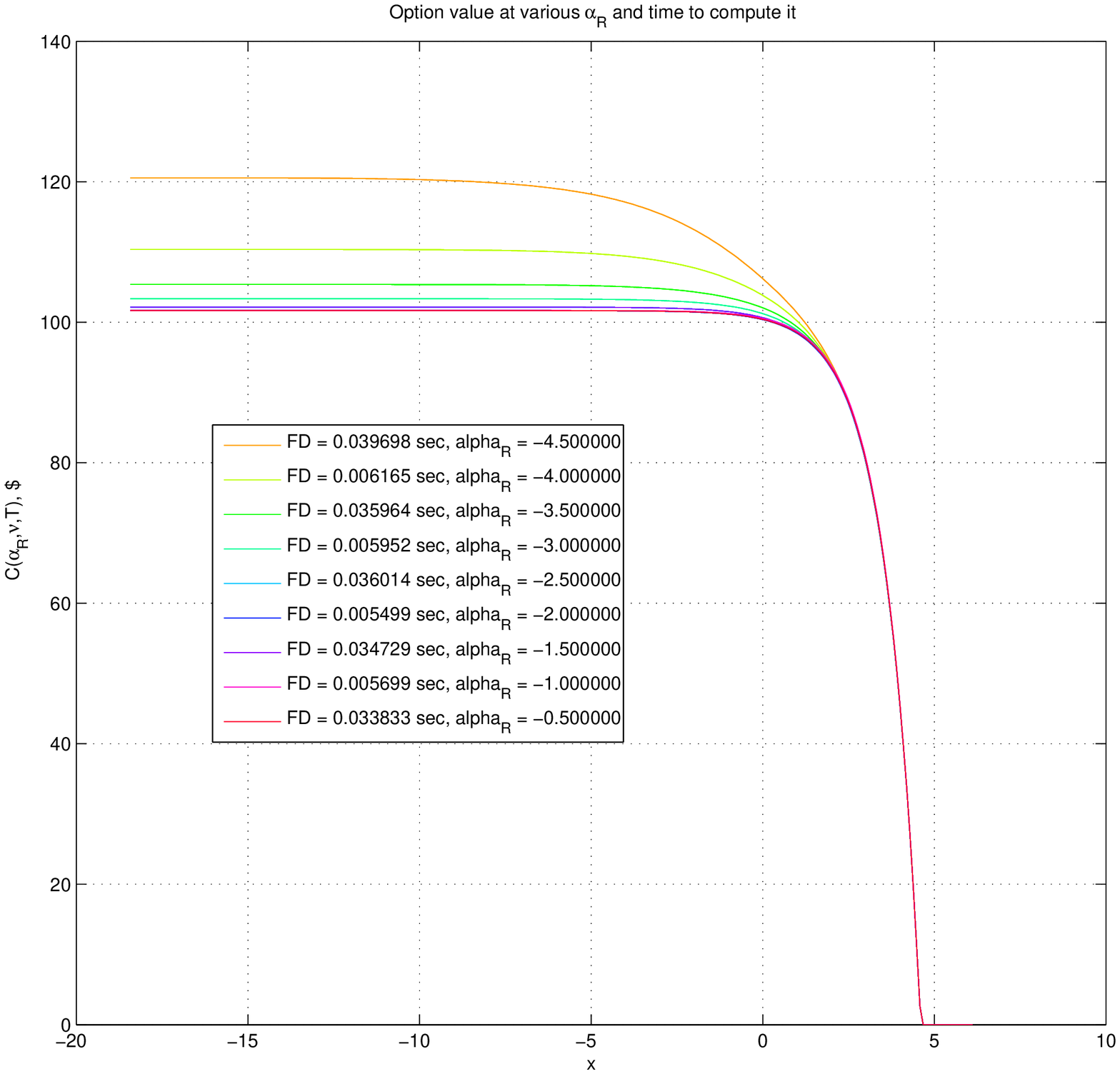}}
\caption{Option values computed using such a scheme as a function of $x$ obtained in the same test}
\label{Calphas}
\end{figure}

\section{General case} \label{Sgc}
If we take a more close look at the propositions \ref{p1} and \ref{p2} we could recognize that the assumption $\alpha \in \mathbb{I}$ could be neglected while both propositions will remain valid. This could be easily seen based on the following equalities

\begin{proposition} \label{p5}
Assume that in the Eq.~(\ref{greenFdef}) $\alpha \in \mathbb{R}, \alpha \le -1$. Then the solution of the Eq.~(\ref{greenFdef}) with respect to $\mathcal{A}^+_x$ is
\bd
\mathcal{A}^+_x = \dfrac{1}{\lambda \Gamma(p+1)}\left(\nu + \fp{}{x}\right)^{p+1} \equiv \dfrac{1}{\lambda \Gamma(p+1)}\left[
\sum_{i=0}^{\infty} C^{p+1}_{i} \nu^{p+1-i} \dfrac{\partial^i}{\partial x^i}\right], \quad p \equiv -(1 + \alpha) \ge 0,
\ed

\ni where $C^{p+1}_{i}$ are the generalized binomial coefficients which could be expressed via Gamma function, and fractional derivatives are understood in the Riemann-Liouville sense (\cite{fracDer})
\end{proposition}

\begin{proof}[{\bf Proof 1}]
Taking Laplace Transform of the expression $\mathcal{A^+} f(x)$ we obtain
\begin{align}
\mathcal{L}_s\left(\mathcal{A}^+_x f(x) \right) &=  \mathcal{L}_s\Biggl\{\dfrac{1}{\lambda \Gamma(p+1)}  \left[\sum_{i=0}^{\infty} C^{p+1}_{i} \nu^{p+1-i} \dfrac{\partial^i}{\partial x^i}\right]f(x)\Biggr\} =  \dfrac{1}{\lambda \Gamma(p+1)} & \left[\sum_{i=0}^{\infty} C^{p+1}_{i} \nu^{p+1-i} s^i \right]  \mathcal{L}_s f(x) \nn \\
&=  \dfrac{1}{\lambda \Gamma(p+1)}(\nu + s)^{p+1} \mathcal{L}_s f(x) \nn
\end{align}

Now, as
\bd
f(x) = \lambda \dfrac{e^{-\nu |x|}}{|x|^{1+\alpha}}\mathbf{1}_{y>0}
\ed

\ni and
\bd
\mathcal{L}_s \Biggl\{ \lambda \dfrac{e^{-\nu |x|}}{|x|^{1+\alpha}}\mathbf{1}_{x>0} \Biggr\} = \lambda \Gamma (p+1) (\nu+s)^{-(1+p)},
\ed

\ni we obtain
\bd
\mathcal{L}_s\left(\mathcal{A}^+_x f(x) \right) = 1 = \mathcal{L}_s \delta(x)
\ed

And thus $\mathcal{A}^+_x f(x) = \delta(x)$
\end{proof}

For the operator $\mathcal{A}^-_x$ the proof is similar.

Another proof is based on a different idea.
\begin{proof}[{\bf Proof 2}]
As it is well known a shift operator in L2 space could be represented as follows
\begin{equation}\label{shift}
    \mathfrak{S}_a = \exp \left( a \dfrac{\partial}{\partial x} \right),
\end{equation}

\ni so
\bd
\mathfrak{S}_a f(x) = f(x+a).
\ed

Therefore, the integrals in the Eq.~(\ref{intX}) could be formally rewritten as
\begin{align}\label{intGen}
\mathcal{A}_1 C(x,t), &\quad \mathcal{A}_1 \equiv
 \int_0^{\infty} \lambda_R \dfrac{e^{-\nu_R |y|}}{|y|^{1+\alpha_R}} \exp \left( y \dfrac{\partial}{\partial x} \right) dy \\
\mathcal{A}_2 C(x,t), &\quad \mathcal{A}_2 \equiv
 \int^0_{-\infty} \lambda_L \dfrac{e^{-\nu_L |y|}}{|y|^{1+\alpha_L}} \exp \left( y \dfrac{\partial}{\partial x} \right) dy \nn
\end{align}

We can compute these integrals assuming that $\partial/ \partial x$ is a constant. This gives
\begin{align}\label{intRes}
\mathcal{A}_1 &= \lambda_R \Gamma(-\alpha_R) \left( \nu_R - \dfrac{\partial}{\partial x} \right)^{\alpha_R}, \quad \mathbb{R} (\alpha) < 0, \mathbb{R}(\nu_R - \partial/ \partial x) > 0 \\
\mathcal{A}_2 &= \lambda_L \Gamma(-\alpha_L)\left( \nu_L + \dfrac{\partial}{\partial x} \right)^{\alpha_L}, \quad \mathbb{R}(\alpha) < 0, \mathbb{R}(\nu_L + \partial/ \partial x) > 0, \nn
\end{align}

\ni where under a real part of differential operator we will understand the real part of the maximum eigenvalue of finite difference matrix which approximates this differential operator (see below).

A simple observation shows that
\bd
\mathcal{A}_1 = \left( \mathcal{A}_x^- \right)^{-1}, \quad \mathcal{A}_2 = \left( \mathcal{A}_x^+ \right)^{-1},
\ed

\ni which finalizes the proof.

\end{proof}

This means that the whole analysis of the previous sections made in the case $\alpha \in \mathbb{I}$ is still valid for arbitrary $\alpha \in \mathbb{R}, \alpha < 0$. Moreover, we could now extend this proof for the whole range of $\alpha < 2$. In order to do that we have to consider the whole integrals in the Eq.~(\ref{pide_init}). This is because in the case of jumps with infinite activity or infinite variation the second and third integrands can not be integrated out, because they do not exist.

If we apply the second transformation to the first equation in the Eq.~(\ref{intX}) the result is given by the following proposition.

\begin{proposition} \label{p6}
The PIDE
\begin{align} \label{orig}
\fp{}{\tau}& C(x,V_R,V_L,\tau) = \nn \\
 &\sqrt{V_R} \int_0^{\infty} \left[C(x+y,V_R,V_L,\tau) - C(x,V_R,V_L,\tau) - \fp{}{x} C(x,V_R,V_L,\tau) (e^y-1)  \right] \lambda_R \frac{e^{-\nu_R |y| }}{|y|^{1+\alpha_R}} dy
\end{align}

is equivalent to PDE
\begin{align} \label{whole}
\fp{}{\tau} C(x,V_R,V_L,\tau) &= \sqrt{V_R} \lambda_R \Gamma(-\alpha_R) \left\{ \left(\nu_R - \dfrac{\partial}{\partial x}\right)^{\alpha_R} - \nu_R^{\alpha_R} + \left[ \nu_R^{\alpha_R} - (\nu_R-1)^{\alpha_R}\right] \dfrac{\partial}{\partial x} \right\}C(x,V_R,V_L,\tau), \nn \\
& \mathbb{R}(\alpha_R) < 2, \ \mathbb{R}(\nu_R - \partial/ \partial x) > 0, \ \mathbb{R}(\nu_R) > 1.
\end{align}

In special cases this equation changes to
\begin{align} \label{whole0}
\fp{}{\tau} C(x,V_R,V_L,\tau) &= \sqrt{V_R}\lambda_R  \left\{ \log(\nu_R) - \log \left(\nu_R - \dfrac{\partial}{\partial x} \right) + \log \left(\dfrac{\nu_R-1}{\nu_R}\right)\dfrac{\partial}{\partial x}  \right\} C(x,V_R,V_L,\tau)\\
& \alpha_R = 0, \mathbb{R}(\nu_R - \partial/ \partial x) > 0, \mathbb{R}(\nu_R) > 1, \nn
\end{align}

\ni and
\begin{align} \label{whole1}
\fp{}{\tau} & C(x,V_R,V_L,\tau) = \sqrt{V_R}\lambda_R \Big\{ - \nu_R \log \nu_R + (\nu_R- \fp{}{x})\log \left(\nu_R-\fp{}{x}\right) \\
&+ \left[\nu_R \log \nu_R - (\nu_R-1) \log (\nu_R-1)\right]\fp{}{x}   \Big\} C(x,V_R,V_L,\tau) \nn \\
& \alpha_R = 1, \mathbb{R}(\partial/ \partial x) < 0, \mathbb{R}(\nu_R) > 1, \nn
\end{align}

\ni where logarithm of the differential operator is defined in a sense of (\cite{logOfDif}).
\end{proposition}

\begin{proof}[{\bf Proof}]
We again use the shift operator introduced in the Eq.~(\ref{shift}) to rewrite the Eq.~(\ref{orig}) as
\begin{align} \label{whole1p}
\fp{}{\tau} & C(x,V_R,V_L,\tau) = \mathcal{B}_1 C(x,V_R,V_L,\tau) \nn \\
& \mathcal{B}_1 \equiv \sqrt{V_R} \int_0^{\infty} \left[ \exp \left( y \dfrac{\partial}{\partial x} \right) - 1 - (e^y-1) \fp{}{x} \right] \lambda_R \frac{e^{-\nu_R |y| }}{|y|^{1+\alpha_R}} dy
\end{align}

Formal integration could be fulfilled if we treat a differential operator $\fp{}{x}$ as a parameter. As it could be verified the result is that given in the Eq.~(\ref{whole}). Same method is used to prove the formulae given  in the special cases $\alpha_R = 0$ and $\alpha_R = 1$.
\end{proof}

Also notice that at $\alpha_R = 0$ from the very beginning the last term in the Eq.~(\ref{orig}) can be moved from the integral to the diffusion part of the Eq.~(\ref{pide_init}) because the remaining kernel converges at $y=0$. If we do so, at this special case the integrated equation transforms to
\begin{align} \label{whole0-1}
\fp{}{\tau} C(x,V_R,V_L,\tau) &= \sqrt{V_R}\lambda_R  \left\{ \log(\nu_R) - \log \left(\nu_R - \dfrac{\partial}{\partial x} \right) \right\} C(x,V_R,V_L,\tau)\\
& \alpha_R = 0, \mathbb{R}(\nu_R - \partial/ \partial x) > 0, \mathbb{R}(\nu_R) > 0, \nn
\end{align}

This form is more useful as we show later when elaborating a numerical method to solve it.
\\

The same approach could be utilized for the second equation in the Eq.~(\ref{intX}), and the result is given by the following proposition.

\begin{proposition} \label{p7}
The PIDE
\begin{align} \label{origM}
\fp{}{\tau}& C(x,V_R,V_L,\tau) = \nn \\
 &\sqrt{V_L} \int_{-\infty}^0 \left[C(x+y,V_R,V_L,\tau) - C(x,V_R,V_L,\tau) - \fp{}{x} C(x,V_R,V_L,\tau) (e^y-1)  \right] \lambda_L \frac{e^{-\nu_L |y| }}{|y|^{1+\alpha_L}} dy
\end{align}

is equivalent to PDE
\begin{align} \label{wholeM}
\fp{}{\tau} C(x,V_R,V_L,\tau) &= \sqrt{V_L} \lambda_L \Gamma(-\alpha_L) \left\{ \left(\nu_L + \dfrac{\partial}{\partial x}\right)^{\alpha_L} - \nu_L^{\alpha_L} + \left[ \nu_L^{\alpha_L} - (\nu_L+1)^{\alpha_L}\right] \dfrac{\partial}{\partial x} \right\}C(x,V_R,V_L,\tau), \nn \\
& \mathbb{R}(\alpha_L) < 2, \ \mathbb{R}(\nu_L + \partial/ \partial x) > 0, \ \mathbb{R}(\nu_L) > 0.
\end{align}

In special cases this equation changes to
\begin{align} \label{whole0M}
\fp{}{\tau} C(x,V_R,V_L,\tau) &= - \sqrt{V_L} \lambda_L  \left\{\log \left(\nu_L + \dfrac{\partial}{\partial x} \right)  - \log(\nu_L) - \log \left(\dfrac{\nu_L+1}{\nu_L}\right)\dfrac{\partial}{\partial x}  \right\} \\
& \alpha_L = 0, \ \mathbb{R}(\nu_L + \partial/ \partial x) > 0, \ \mathbb{R}(\nu_L) > 0, \nn
\end{align}

\ni and
\begin{align} \label{whole1M}
\fp{}{\tau} & C(x,V_R,V_L,\tau) = \sqrt{V_L} \lambda_L \Big\{ - \nu_L \log \nu_L \\
&+ \left[\nu_L \log \nu_L - (\nu_L+1) \log (\nu_L+1)\right]\fp{}{x} + (\nu_L+\fp{}{x})\log \left(\nu_L+\fp{}{x}\right)  \Big\} C(x,V_R,V_L,\tau) \nn \\
& \alpha_R = 1, \ \mathbb{R}(\partial/ \partial x) < 0, \ \mathbb{R}(\nu_L) > 0, \nn\end{align}

\ni where logarithm of the differential operator is defined in a sense of (\cite{logOfDif}).
\end{proposition}

\begin{proof}[{\bf Proof}]
The proof is similar to that given in the Proposition~\ref{p6}.
\end{proof}

Again at $\alpha_L = 0$ we can move out the last term in the Eq.~(\ref{orig}) from the integral to the diffusion part of the Eq.~(\ref{pide_init}) because the remaining kernel converges at $y=0$. If we do so, at this special case the integrated equation transforms to
\begin{align} \label{whole0M-1}
\fp{}{\tau} C(x,V_R,V_L,\tau) &= - \sqrt{V_L} \lambda_L  \left\{\log \left(\nu_L + \dfrac{\partial}{\partial x} \right)  - \log(\nu_L) \right\} \\
& \alpha_L = 0, \ \mathbb{R}(\nu_L + \partial/ \partial x) > 0, \ \mathbb{R}(\nu_L) > 0, \nn
\end{align}

We will use this form later when elaborating a numerical method to solve this equation.

It is important to underline that the integration in the Proposition \ref{p6} for positive jumps could be done if $\mathbb{R}(\nu_R) > 1$ while in the Proposition \ref{p7} for negative jumps - if $\mathbb{R}(\nu_L) > 0$. In the special cases $\alpha_R = 1$ this limit could be extended to $\mathbb{R}(\nu_R) > 0$, however it gives rise to a complex values of the coefficients in the rhs of the Eq.~(\ref{whole1M}). Therefore, we keep the above constraint $\mathbb{R}(\nu_R) > 1$ unchanged in this case as well.

Similar representations were obtained first in \cite{BL2002} and later in \cite{Cartea2007} using a characteristic function approach. For instance, the latter authors considered several L\'evy processes with known characteristic function, namely LS, CGMY or KoBoL. Then using Fourier transform they managed to convert the governing PIDE (same type as the Eq.~(\ref{pide_init}) but for the Black-Scholes model with jumps) to a fractional PDE. In their notation our operator $\mathcal{A}_1$ is represented as
\begin{equation}\label{CarteaA1}
    \mathcal{A}_1 \propto (-1)^{\alpha_R} e^{\nu_R} \ _x\mathbb{D}_\infty^{\alpha_R} \left( e^{-\nu_R} C(x,t) \right),
\end{equation}

and operator $\mathcal{A}_2$ as
\begin{equation}\label{CarteaA2}
    \mathcal{A}_2 \propto e^{\nu_L} \ _\infty\mathbb{D}_x^{\alpha_L} \left( e^{-\nu_L} C(x,t) \right),
\end{equation}

So to compare we have to note that aside of the different method of how to derive these equations our main contribution in this paper is:
\begin{enumerate}
\item Special cases $\alpha_r=0,1, \alpha_l=0,1$ are not considered in \cite{Cartea2007}. In \cite{BL2002} a corresponding characteristic function of the KoBoL process was obtained in all cases for $\alpha \le 1$. However, the authors did not consider numerical solution of the fractional PDE. In this paper we derive a fractional PDE for all $\alpha < 2$ and propose a numerical method for their solution.

\item We proposed the idea of solving FPDE with real $\alpha_R \le 0, \alpha_L \le 0$ by using interpolation between option prices computed for the closest integer values of $\alpha_R, \alpha_L$. For the latter we first used to transform the fractional equation into a pseudo-parabolic equation. Then for the solution of this PPDE an efficient FD scheme is constructed that results in LU factorization of the band matrix.

\item Also jumps up and down are considered separately so the model in use (SSM) is slightly different from the model considered in \cite{Cartea2007}.

\item In \cite{Cartea2007} a Crank-Nicolson type numerical scheme was proposed to solve the obtained FPDE in time while discretization in space was done using the Grunwald-Letnikov  approximation which is of the first order in space. Here for fractional equations with $2 > \alpha_R > 0, 2 > \alpha_L > 0$ we obtain the solution using our new scheme which preserves the second order approximation in time and space.

\item As it is known from recent papers (\cite{AbuSaman2007, MeerschaertTadjeran2004, Tadjeran2006, MeerschaertTadjeran2006, Sousa2008}), a standard Grunwald-Letnikov approximation leads to unconditionally unstable schemes. To improve this a shifted Grunwald-Letnikov approximation was proposed which allows construction of the unconditionally stable scheme of the first order in space. \footnote{A second order approximation could in principle be constructed as well, however resulting in a massive calculation of the coefficients. That probably stopped the scientists to further elaborate this approach.} Here we use a different approach to derive the unconditionally stable scheme of higher order.

\item We show that when considering jumps with finite activity and finite variation despite it is a common practice to integrate out all L\'evy compensators in the  Eq.~(\ref{pide_init}) in the integral terms  this breaks the stability of the scheme at least for the fractional PDE. Therefore, in order to construct the unconditionally stable scheme one must keep some other terms under the integrals. To resolve this in Cartea (2007) the authors were compelled to change their definition of the fractional derivative (see below).

\item Our approach could be easily generalized for a time-dependent L\'evy density.
\end{enumerate}

\section{Numerical method} \label{Snum}
Let us consider a general case which is given by the Eq.~(\ref{whole}) and Eq.~(\ref{wholeM}) \footnote{In principal one can eliminate special cases when one of the following conditions is valid $\alpha_R =0, \alpha_L = 0, \alpha_R = 1, \alpha_L = 1$, by just substituting, say $\alpha_R = \epsilon << 1$ instead of $\alpha_R = 0$, $\alpha_R = 1 + \epsilon$ instead of $\alpha_R = 1$ etc}. We first discuss how to construct an unconditionally stable scheme of the second order in space and second or higher order in time. Then we consider some peculiarities of implementation of the derived finite difference schemes.

\subsection{Case $\alpha_R = 0$ or $\alpha_L = 0$.} This extreme case corresponds to the familiar Variance Gamma model.
In this case the integrals in the Eq.~(\ref{orig}) and Eq.~(\ref{origM}) exist if we keep just first two terms under the integral. Therefore we could integrate out the last term  $R \propto \fp{}{x} C(x,\tau) (e^y-1)$. This term then will become a part of the convection part of the total PIDE and therefore we will not consider it here, assuming that we use a splitting technique and know how to solve the remaining convection-diffusion equation.

Then the Eq.~(\ref{whole}) could be written in the form of the Eq.~(\ref{operEq11}) with
\begin{eqnarray} \label{Balpha0}
    \mathcal{B}_R &=&  \sqrt{V_R}\lambda_R  \left\{ \log(\nu_R) - \log \left(\nu_R - \dfrac{\partial}{\partial x} \right)\right\} \\
    \mathcal{B}_L &=&  \sqrt{V_L}\lambda_L  \left\{ \log(\nu_L)  - \log \left(\nu_L+ \dfrac{\partial}{\partial x} \right) \right\} \nn
\end{eqnarray}

Therefore, integrating it we obtain an explicit form of the Eq.~(\ref{solOper2})
\begin{eqnarray}\label{solAlpha0}
C^{k+1}(x) &=& \left(1 - \dfrac{1}{\nu_R}\dfrac{\partial}{\partial x} \right)^{-m}   C^k(x), \quad m   = \sqrt{V_R}\lambda_R \theta > 0,  \\
C^{k+1}(x) &=& \left(1 + \dfrac{1}{\nu_L}\dfrac{\partial}{\partial x} \right)^{-m} C^k(x), \quad m   = \sqrt{V_L}\lambda_L \theta,  \nn
\end{eqnarray}

In practical computation of the rhs operators we exploit a modification of our interpolation method which was described above. First, note that typical values of $\lambda_R, \lambda_L$ as well as $V_R, V_L$ are limited, i.e. normally $\lambda_R < M, \lambda_L < M, V_R < M, V_L < M$ where M could be chosen in the range, say 3-5. Second, if we solve a general jump-diffusion equation using some kind of splitting methods, the time step of integration $\theta$ in the Eq.~(\ref{solAlpha0}) is determined by the time step used at the integration of the diffusion part. This means that $\theta$ is usually small. Therefore, it is pretty reasonable to assume that in the Eq.~(\ref{solAlpha0}) $m < 2$. Next, as follows from the definition of the fractional derivatives, the operators in the Eq.~(\ref{solAlpha0}) are continuous in $m$. Therefore, we could solve the Eq.~(\ref{solAlpha0}) for $m=0,1,2$ and then use quadratic interpolation to get the solution given the real value of $m$, and the condition $m < 2$. Note, that $m=0$ is a trivial case so the solution $C^{k+1}(X) = C^k(x)$ is already known.

Note a choice of $m=-1$. On the one hand this is very attractive because then the solution of the Eq.~(\ref{solAlpha0}) is already found. On the other hand at $m < 0$ the scheme in the Eq.~(\ref{solAlpha0}) becomes explicit which breaks its unconditional stability. Apparently the best one can achieve in this case is to use a central difference approximation for the first derivative. Then it is possible to show that all eigenvalues of the rhs matrix have their real value equal to one. Thus the stability of the scheme is questionable.

We now construct a stable FD scheme to solve the first equation in the Eq.~(\ref{solAlpha0}). Similar to what was already discussed in the previous section a forward second order approximation of the first derivative has to be chosen. Then the eigenvalues of the discrete operator $\left(1 - \dfrac{1}{\nu_R}\dfrac{\Delta}{\Delta x} \right)^{-m}$ are
\begin{equation}\label{alpha0lambda}
    \zeta = \left(1 + \dfrac{3}{2h \nu_R}\right)^{-m}.
\end{equation}

We need to guarantee that $\| \left(1 - \dfrac{1}{\nu_R}\dfrac{\Delta}{\Delta x} \right)^{-m} \| < 1$. Thus, if $\nu_R < 1$ this FD scheme is stable at $h < 3/[2(1-\nu_R)]$, and if $\nu_R \ge 1$ - it is unconditionally stable. As follows from the Proposition Eq.~(\ref{p6}) $\mathbb{R}(\nu_R) > 1$, therefore the scheme is unconditionally stable.

After this discretization the matrix of the lhs operator becomes one-sided tridiagonal if $m=1$, and one-sided pentadiagonal if $m=2$. Therefore this equation can be efficiently solved with the total complexity $O(N(2m+1))$.

To preserve monotonicity of the solution for the second equation in the Eq.~(\ref{solAlpha0}) a backward second order approximation of the first derivative has to be chosen. This approximation was also already introduced in the previous section. Then $C^{k+1}(x,m)$ can be computed as a product $A_m \cdot C^k(x)$, where $A_m$ is a band matrix with $2m+1$ diagonals. So the complexity of this is also $O(N(2m+1)$.

Based on these results we extend our numerical test described in the previous section to the case $\alpha_R = 0$. However, to preserve convergence of the integral now instead of the Eq.~(\ref{numEq}) we have to use the extended equation
\begin{equation} \label{numEq2}
\fp{}{\tau} C(x,\tau) = \int_0^{\infty} \left[C(x+y,\tau) - C(x,\tau)\right]\lambda_R \dfrac{e^{-\nu_R |y|}}{|y|^{1+\alpha_R}} dy,
\quad \alpha_R < -1
\end{equation}

We again compare the FFT solution of the Eq.~(\ref{numEq2}) with that obtained based on our method.

\paragraph{FFT.} It should be underlined that the presented simple FFT algorithm completely loses its accuracy when $\alpha_R \rightarrow 0$. Therefore, instead of $\alpha_R = 0$ we will chose real $\alpha_R= -0.5$. We again define a uniform grid in the domain $(-x_*,x_*)$ which contains $N$ points: $x_1 = -x_*, x_2,...x_{N-1}, x_N = x_*$ such that $x_i - x_{i-1} = h, i=2...N$. We then approximate the integral in the rhs of the Eq.~(\ref{numEq2}) with the first order of accuracy in $h$ as
\begin{align} \label{numEq2Appr}
\int_0^\infty & \left[C(x+y,\tau) - C(x,\tau)\right]\lambda_R \dfrac{e^{-\nu_R |y|}}{|y|^{1+\alpha_R}} dy =
h \sum_{j=1-i}^{N-i} C_{i+j}(\tau) f_j - C(x,\tau) \lambda_R \nu_R^{\alpha_R} \Gamma(-\alpha_R), \nn \\
f_j &\equiv \lambda_R \dfrac{e^{-\nu_R |x_j|}}{|x_j|^{1+\alpha_R}} + O(h^2)
\end{align}

The matrix-vector product in the lhs of the Eq.~(\ref{numEq2Appr}) is computed using FFT as it was described in the previous section.

\paragraph{FD.} We solve the Eq.~(\ref{numEq2}) using interpolation in $\alpha_R$ between the points $\alpha_R = 0,-1,-2,-3$. At $\alpha_R = 0$ we use the FD scheme in the Eq.~(\ref{solAlpha0}). At $\alpha_R < 0$ we again use our approach of construction of the pseudo-parabolic equations (see propositions \ref{p3}, \ref{p4}), and instead of the Eq.~(\ref{operEq11}) now obtain
\begin{equation}\label{operEq111}
   \fp{}{\tau}C(x,t) =  \mathcal{B} C(x,t), \quad \mathcal{B} \equiv \dfrac{1}{2} (\mathcal{A}^-_x)^{-1} - \lambda_R \nu_R^{\alpha_R}\Gamma(-\alpha_R).
\end{equation}

Further we use the Crank-Nicolson scheme Eq.~(\ref{discrete}) which now reads
\begin{equation}\label{discrete2terms}
\left( \left[1 + \dfrac{1}{2}\lambda_R \nu_R^{\alpha_R}\Gamma(-\alpha_R)\theta\right]\mathcal{A}^-_x - \dfrac{1}{4} \theta\right) C^{k+1}(x) =  \left(\left[1 - \dfrac{1}{2}\lambda_R \nu_R^{\alpha_R}\Gamma(-\alpha_R)\theta\right]\mathcal{A}^-_x + \dfrac{1}{4} \theta\right) C^k(x).
\end{equation}

The stability analysis could be provided similar to what we did in the previous sections. Again it is easy to show that the forward one-sided approximation of the operator $\mathcal{A}^-_x$ given in the Eq.~(\ref{forward1der}) guarantees the unconditional stability of the above scheme.

\begin{figure}[t!]
\vspace{0.1 in}
\hspace{0.5 in}
\centering
\fbox{\includegraphics[totalheight=3.5in]{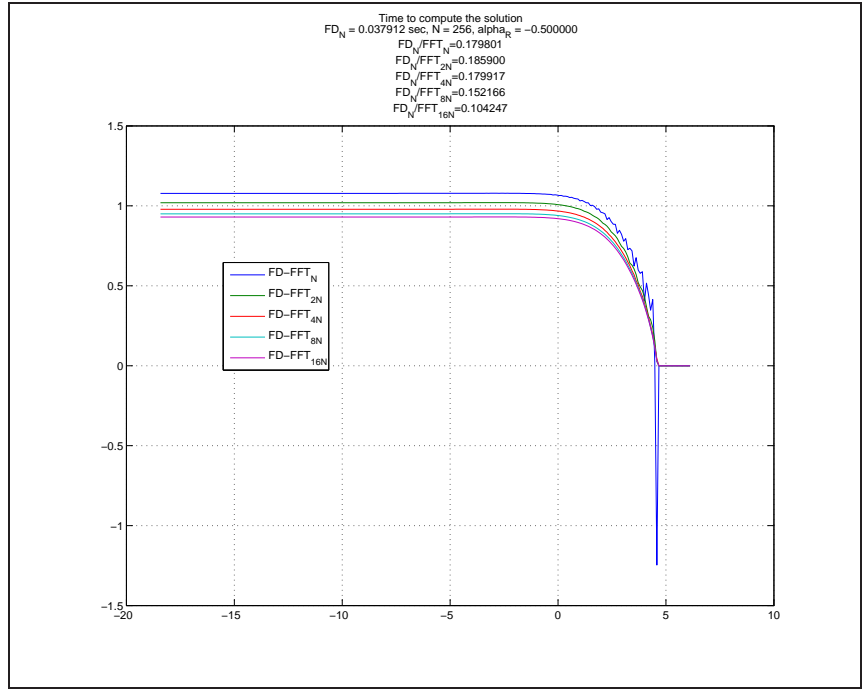}}
\caption{Difference (FD-FFT) in solutions of the Eq.~(\ref{numEq1}) obtained using our finite-difference method (FD) and an explicit Euler scheme in time where the jump integral is computed using FFT. $\alpha_R = -0.5$.}
\label{Fig0}
\end{figure}

\paragraph{Comparison.} The results of this test are given in Fig.~\ref{Fig0}. This could be compared with the results presented in Fig.~\ref{Fig0_5}. The difference is that now instead of extrapolation we use interpolation, because we are able to solve our test problem numerically at $\alpha_R = 0$. Surprisingly the difference in the FFT and FD solutions slightly increases in case of interpolation. The FD solution is still faster than the FFT, and as follows from the above analysis - more accurate.

\subsection{Case $\alpha_R = 1, \alpha_L = 1$.}

This is a case of jumps with infinite variation and infinite activity. Therefore we have to keep the whole integrals in the Eq.~(\ref{orig}) and Eq.~(\ref{origM}), i.e. in each integral we can not integrate the last term out because otherwise the integral does not converge.

Let us remind that as follows from the Proposition \ref{p6} in this  case the original PIDE Eq.~(\ref{orig}) is equivalent to the PIDE
\begin{align} \label{whole1-1}
\fp{}{\tau} & C(x,V_R,V_L,\tau) = \sqrt{V_R}\lambda_R \Big\{ - \nu_R \log \nu_R + (\nu_R- \fp{}{x})\log \left(\nu_R-\fp{}{x}\right) \\
&+ \left[\nu_R \log \nu_R - (\nu_R-1) \log (\nu_R-1)\right]\fp{}{x}   \Big\} C(x,V_R,V_L,\tau) \nn \\
& \mathbb{R}(\partial/ \partial x) < 0, \mathbb{R}(\nu_R) > 1, \nn
\end{align}

\ni while from Proposition \ref{p7} the PIDE Eq.~(\ref{whole1M}) is equivalent to the PIDE
\begin{align} \label{whole1M-1}
\fp{}{\tau} & C(x,V_R,V_L,\tau) = \sqrt{V_L} \lambda_L \Big\{ - \nu_L \log \nu_L \\
&+ \left[\nu_L \log \nu_L - (\nu_L+1) \log (\nu_L+1)\right]\fp{}{x} + (\nu_L+\fp{}{x})\log \left(\nu_L+\fp{}{x}\right)  \Big\} C(x,V_R,V_L,\tau) \nn \\
& \mathbb{R}(\partial/ \partial x) < 0, \ \mathbb{R}(\nu_L) > 0. \nn
\end{align}

For the following we need to prove the following Proposition.

\begin{proposition} \label{p8}
The following identity holds
\begin{align}
- \nu_R \log \nu_R &+ (\nu_R- \fp{}{x})\log \left(\nu_R-\fp{}{x}\right) + \left[\nu_R \log \nu_R - (\nu_R-1) \log (\nu_R-1)\right]\fp{}{x}   \nn \\
&= \int_\nu^\infty \Biggl\{ \log \nu_R - \log \left( \nu_R - \fp{}{x}\right) + \left(\log \dfrac{\nu_R-1}{\nu_R}\right)\fp{}{x} \Biggr\} d \nu
\end{align}
\end{proposition}

\begin{proof}[{\bf Proof}]
To prove this we one have to note that
\begin{equation}\label{intAlpha1}
  \int^\infty_\nu \dfrac{e^{-\nu_R |y|}}{|y|^{1+\alpha_R}} d\nu = \dfrac{e^{-\nu_R |y|}}{|y|^{2+\alpha_R}},
\end{equation}

\ni and then use Proposition~\ref{p6} with $\alpha_R=0$.
\end{proof}

In a similar way we can prove the following proposition
\begin{proposition} \label{p9}
\begin{align}
- \nu_L \log \nu_L &+ \left[\nu_L \log \nu_L - (\nu_L+1) \log (\nu_L+1)\right]\fp{}{x} + (\nu_L+\fp{}{x})\log \left(\nu_L+\fp{}{x}\right)   \nn \\
&= \int_\nu^\infty \Biggl\{ \log(\nu_L) - \log \left(\nu_L + \dfrac{\partial}{\partial x} \right)  + \log \left(\dfrac{\nu_L+1}{\nu_L}\right)\dfrac{\partial}{\partial x}  \Biggr\} d \nu
\end{align}
\end{proposition}
$\Box$

These two identities gives us an idea of how to construct a FD numerical method for solving the Eq.~(\ref{whole1-1}) and Eq.~(\ref{whole1M-1}).
First we rewrite the Eq.~(\ref{whole1-1}) and Eq.~(\ref{whole1M-1}) in the form
\begin{align}\label{newEqAlpha1}
\fp{}{\tau} C(x,V_R,V_L,\tau) &= \mathbb{L}_R C(x,V_R,V_L,\tau) \\
\fp{}{\tau} C(x,V_R,V_L,\tau) &= \mathbb{L}_L C(x,V_R,V_L,\tau) \nn \\
\mathbb{L}_R &\equiv \sqrt{V_R}\lambda_R \int_\nu^\infty \Biggl\{ \log (\nu_R) - \log \left( \nu_R - \fp{}{x}\right) + \left(\log \dfrac{\nu_R-1}{\nu_R}\right)\fp{}{x} \Biggr\} d \nu \nn \\
\mathbb{L}_L &\equiv \sqrt{V_L}\lambda_L \int_\nu^\infty \Biggl\{ \log(\nu_L) - \log \left(\nu_L + \dfrac{\partial}{\partial x} \right) + \log \left(\dfrac{\nu_L+1}{\nu_L}\right)\dfrac{\partial}{\partial x}  \Biggr\} d \nu \nn
\end{align}

We already know how to solve these equations if the operators $\mathbb{L}_R$ and $\mathbb{L}_L$ do not contain the integrals. We want to utilize this approach by proceeding with the following steps.

\paragraph{Step 1.} First we truncate the upper limit in the integral to some $\nu_*$. This could be done because the integral in the Eq.~(\ref{newEqAlpha1}) is well-defined and at $\nu_R \rightarrow \infty$ the integral kernel tends to zero as
\begin{equation}\label{limit}
  \lim_{\nu_R \rightarrow \infty} \mathbb{L}_R C(x,V_R,V_L,\tau) = \sqrt{V_R}\lambda_R \dfrac{1}{2 \nu_R^2} \left(-\fp{}{x} + \sop{}{x}\right) + O(1/\nu_R^3)
\end{equation}

 At the interval $(\nu, \nu_*)$ we approximate the integral in $\nu$ using some quadrature formula, for instance, the well-known Simpson formula (higher-order approximations of even adaptive quadratures could definitely be used as well). So we partition the interval $(\nu, \nu_*)$ into an even number of intervals $M$ all of the same width $h = (\nu_* - \nu)/M$. Then operators in the Eq.~(\ref{newEqAlpha1}) transform to
\begin{align}\label{newEqAlpha2}
\mathbb{L}_R &\equiv  \sum_{i=0}^M \mathbb{L}_{i,R}, \qquad \mathbb{L}_L \equiv  \sum_{i=0}^M \mathbb{L}_{i,L} \\
\mathbb{L}_{i,R} &= a_i \sqrt{V_R}\lambda_R \dfrac{\nu_* - \nu}{3M} \Biggl\{ \log (\nu_{i,R}) - \log \left( \nu_{i,R} - \fp{}{x}\right) + \left(\log \dfrac{\nu_{i,R}-1}{\nu_{i,R}}\right)\fp{}{x} \Biggr\} \nn \\
\mathbb{L}_{i,L} &\equiv a_i \sqrt{V_L}\lambda_L \dfrac{\nu_* - \nu}{3M} \Biggl\{ \log(\nu_{i,L}) -\log \left(\nu_{i,L} + \dfrac{\partial}{\partial x} \right) + \log \left(\dfrac{\nu_{i,L}+1}{\nu_{i,L}}\right)\dfrac{\partial}{\partial x}  \Biggr\},  \nn \\
a_i &= 1, \quad i=0,M, \qquad a_i = 2, \quad i=2,4...M-2, \qquad a_i = 4, \quad i=1,3...M-1. \nn
\end{align}

\paragraph{Step 2.} Each operator in the Eq.~(\ref{newEqAlpha2}) is a sum of $M$ operators which commute with each other. Therefore, the solution of the Eq.~(\ref{newEqAlpha1}) reads
\begin{align} \label{splitting}
  C(x,V_R,V_L,\tau) &= \exp\left[\sum_{i=0}^M \mathbb{L}_{i,R} \tau \right]  C(x,V_R,V_L,0) = \prod_{i=1}^M e^{\mathbb{L}_{i,R} \tau}C(x,V_R,V_L,0) \\
    C(x,V_R,V_L,\tau) &= \exp\left[\sum_{i=0}^M \mathbb{L}_{i,L} \tau \right]  C(x,V_R,V_L,0) = \prod_{i=1}^M e^{\mathbb{L}_{i,L} \tau}C(x,V_R,V_L,0) \nn
\end{align}

Using a splitting technique (see, for instance, \cite{LanserVerwer, Yoshida}) we can represent this equation in the form
\begin{align} \label{split2}
C_1(x,V_R,V_L,\theta) &= e^{\mathbb{L}_{1,R} \tau}C(x,V_R,V_L,0) \\
C_2(x,V_R,V_L,\theta) &= e^{\mathbb{L}_{2,R} \tau}C_1(x,V_R,V_L,\theta) \nn \\
....................&.......................... \nn \\
C_M(x,V_R,V_L,\theta) &= e^{\mathbb{L}_{M,R} \tau}C_{M-1}(x,V_R,V_L,\theta) \nn \\
C(x,V_R,V_L,\theta) &= C_{M}(x,V_R,V_L,\theta) \nn
\end{align}

\ni and similarly for the operator $\mathbb{L}_L$.

\paragraph{Step 3.} Each equation in the Eq.~(\ref{split2}) is very similar to that corresponding to the case $\alpha = 0$ (see the previous section).The only difference is that the operators $\mathbb{L}_{i,R}$ now contain an extra term $\mathbb{L}_{3,i,R} = \left(\log \dfrac{\nu_{i,R}-1}{\nu_{i,R}}\right)\fp{}{x}$, and the operators  $\mathbb{L}_{i,L}$ now contain an extra term $\mathbb{L}_{3,i,L} = \left(\log \dfrac{\nu_{i,L}+1}{\nu_{i,L}}\right)\fp{}{x}$. We can apply splitting to these operators similar to as we did in the above. Further by analogy with what was already discussed in the previous sections devoted to P\'ade approximations, these terms $e^{\mathbb{L}_{3,i,R}}\theta$ and $e^{\mathbb{L}_{3,i,L}}\theta$ could be approximated with the second order of accuracy in $\theta$ by using the Eq.~(\ref{Pade1}). Finally, each equation in the Eq.~(\ref{split2}) reads
\begin{align}\label{finAlpha1}
C_{-1}^{k+1}(x) &= C^{k}(x) \\
C_{i*}^{k+1}(x) &= \dfrac{1 +  \dfrac{m_i}{2}\mathbb{L}_{3,i,R}\theta}{1 -  \dfrac{m_i}{2} \mathbb{L}_{3,i,R}\theta}  C_{i-1}^k(x) \nn \\
C_i^{k+1}(x) &= \left(1 - \dfrac{1}{\nu_{1,R}}\dfrac{\partial}{\partial x} \right)^{-m_i} C_{i*}^k(x), \quad i=0,...,M, \qquad m_i \equiv a_i \sqrt{V_R}\lambda_R \dfrac{\nu_* - \nu}{3M} \theta \nn \\
C^{k+1}(x) &= C_M^{k+1}(x) \nn
\end{align}

We can chose the number $M$ to guarantee that the value of $m_i$ is less than 2 and then use interpolation solving the above equations at $m_i = 0,1,2$.

Similar scheme could be constructed for the operator $\mathbb{L}_{L}$, which reads
\begin{align}\label{finAlpha2}
C_{-1}^{k+1}(x) &= C^{k}(x) \\
C_{i*}^{k+1}(x) &= \dfrac{1 +  \dfrac{m_i}{2}\mathbb{L}_{3,i,L}\theta}{1 -  \dfrac{m_i}{2} \mathbb{L}_{3,i,L}\theta}    C_{i-1}^{k+1}(x) \nn \\
C_i^{k+1}(x) &= \left(1 + \dfrac{1}{\nu_{i,L}}\dfrac{\partial}{\partial x} \right)^{-m_i} C_{i*}^{k+1}(x), \quad i=0,...,M, \qquad m_i \equiv a_i \sqrt{V_L}\lambda_L \dfrac{\nu_* - \nu}{3M} \theta \nn \\
C^{k+1}(x) &= C_M^{k+1}(x) \nn
\end{align}

\paragraph{Step 4.} To construct an unconditionally stable scheme in $x$ we have to chose approximation for the first derivative in the Eq.~(\ref{finAlpha1}). If we rewrite this equation in the form
\begin{align}\label{finAlpha3}
C_{-1}^{k+1}(x) &= C^{k}(x) \\
\left[1 - \dfrac{m_i \theta}{2}\left(\log \dfrac{\nu_{i,R}-1}{\nu_{i,R}}\right)\fp{}{x}\right] C_{i*}^{k+1}(x) &= \left[1 + \dfrac{m_i \theta}{2}\left(\log \dfrac{\nu_{i,R}-1}{\nu_{i,R}}\right)\fp{}{x}\right]  C^k_{i-1}(x),  \quad i=0,...,M \nn \\
\left(1 - \dfrac{1}{\nu_{i,R}}\dfrac{\partial}{\partial x} \right)^{m_i}C^k_{i}(x) &= C^k_{i*}(x) \nn \\
C^{k+1}(x) &= C_M^{k+1}(x) \nn
\end{align}

\ni it becomes obvious that the derivative in the second equation in the Eq.~(\ref{finAlpha3}) should be approximated by using a backward one-sided second order divided difference. For the derivative in the third equation one has to use a forward approximation.

Similarly we rewrite the Eq.~(\ref{finAlpha2}) in the form
\begin{align}\label{finAlpha4}
C_{-1}^{k+1}(x) &= C^{k}(x) \\
\left[1 - \dfrac{m_i \theta}{2}\left(\log \dfrac{\nu_{i,L}+1}{\nu_{i,L}}\right)\fp{}{x}\right]
C_{i*}^{k+1}(x) &= \left[1 + \dfrac{m_i \theta}{2}\left(\log \dfrac{\nu_{i,L}+1}{\nu_{i,L}}\right)\fp{}{x}\right]  C^k_{i-1}(x),  \quad i=0,...,M \nn \\
\left(1 + \dfrac{1}{\nu_{i,L}}\dfrac{\partial}{\partial x} \right)^{m_i}C_i^{k+1}(x) &= C^k_{i*}(x)\nn \\
C^{k+1}(x) &= C_M^{k+1}(x) \nn
\end{align}

\ni and use a forward approximation for the derivative in the second equation in the Eq.~(\ref{finAlpha4}) and the backward approximation in the third equation.

The matrix in the rhs of the second equation in the Eq.~(\ref{finAlpha4}) is upper tridiagonal. The matrix in the rhs of the third equation in the Eq.~(\ref{finAlpha4}) is lower tridiagonal at $m_i=1$ and lower pentadiagonal at $m_i=2$. The total complexity of the algorithm as compared with the case $\alpha=0$ is: one extra equation at each step, $M$ steps instead of just one in the case $\alpha=0$. Therefore, using the results given in Fig.~\ref{Fig0} we can expect that at $M = 30$ this algorithm is about 3 times slower than the FFT. On the other hand it provides the second order approximation in both space and time, and does not require to re-interpolate the FFT results to the FD grid which was previously used to find solution for the diffusion part of the original PIDE.

To verify this we provided two numerical experiments. In the first experiment $\nu_*$ varied while $h = (\nu_* - \nu_R)/M$ was chosen to be constant. At $\nu_* = 5$ we chose $M=30$. The other parameters are same as in the previous numerical experiments reported in the above. This results are presented in Fig.~\ref{varNuStar}.
\begin{figure}[t!]
\vspace{0.1 in}
\hspace{0.5 in}
\centering
\fbox{\includegraphics[totalheight=3.5in]{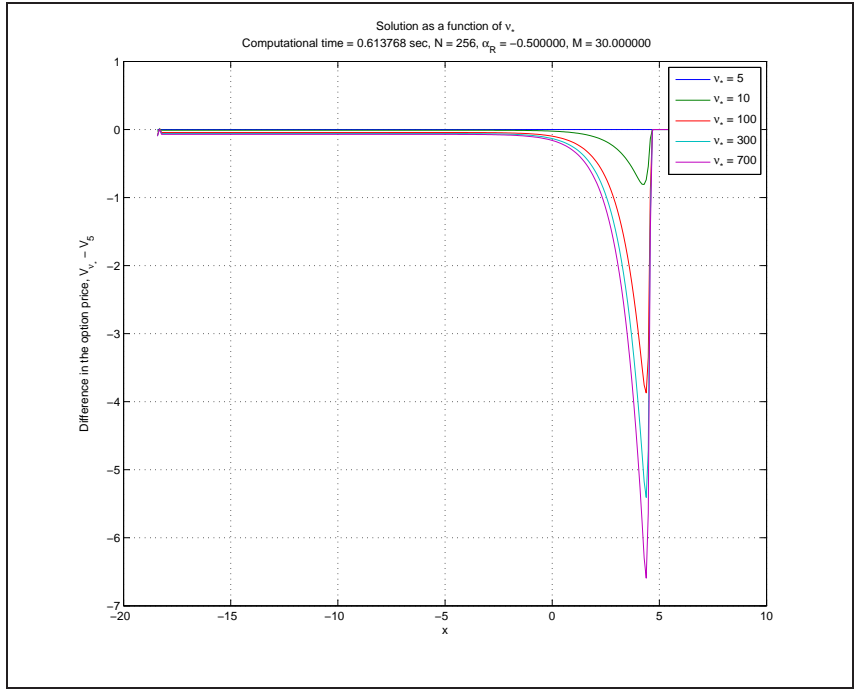}}
\caption{Difference in solutions of the Eq.~(\ref{finAlpha3}) obtained at various $\nu_*$ and that at $\nu_* = 5$ at $M=30$ and $\alpha_R = 1$.}
\label{varNuStar}
\end{figure}

\begin{figure}[t!]
\vspace{0.1 in}
\hspace{0.5 in}
\centering
\fbox{\includegraphics[totalheight=3.5in]{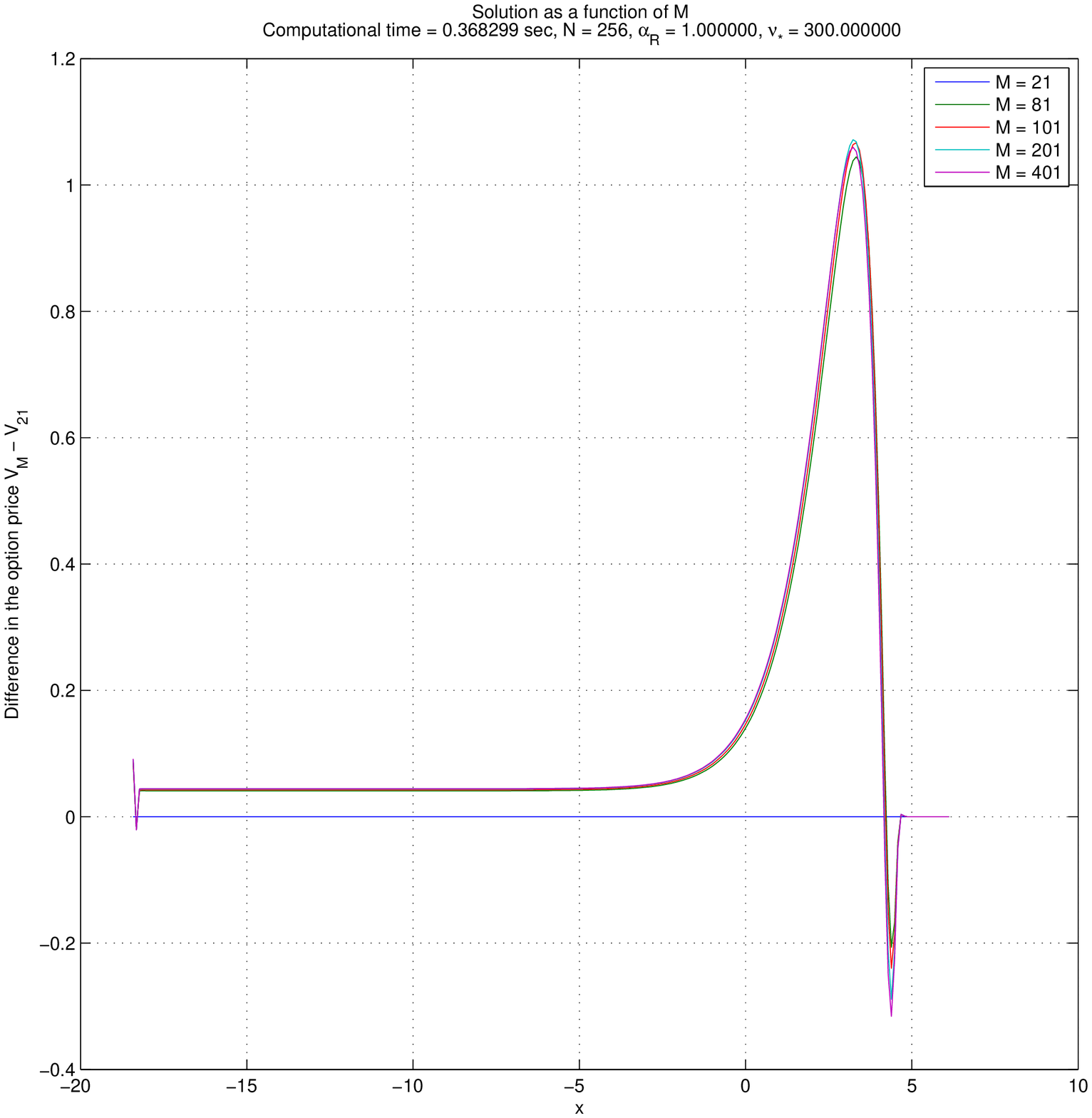}}
\caption{Difference in solutions of the Eq.~(\ref{finAlpha3}) obtained at various $M$ and that at $M = 21$ at $\nu_*=300$ and $\alpha_R = 1$.}
\label{varM}
\end{figure}

The computational time rawly increases by the factor $M/2$, i.e. for $M=30$ it is almost same as for the corresponding FFT. It is seen that an appropriate value of $\nu_*$ should be more than 300.

In the second experiment we fixed the value $\nu_* = 300$ and varied $M$ to see at which $M$ one could expect to get convergency. These results are presented in Fig~\ref{varM}. As it is seen $M=80$ seems to be sufficient to obtain the convergency. The computational time in the case $M=81$ is 1.4 sec which if compared with that given in the Fig~\ref{Fig0} is 3.6 times more than that for the FFT. Thus, in this case our algorithm is almost 4 times slower than the FFT. As it was already mentioned this could be compensated a) by the second order of accuracy in space and time, and b) no need for re-interpolation of the FFT results to the FD grid. One more advantage is that we don't need to treat the point $y=0$ in a special way as it was done, say in \cite{ContVolchkova2003}.

Note, that as we use $M$ steps in the splitting scheme, the error in time becomes $O(M\theta^2)$ that could kill the second order of approximation. Therefore, for instance, in the Eq.~(\ref{finAlpha1}) it is better to use a third order approximation in time (see the Eq.~(\ref{Pade21})). Accordingly the second equation in the  Eq.~(\ref{finAlpha3}) will become
\begin{align}\label{3rdAppr}
\left[1 - \dfrac{2p_i\theta}{3} + \dfrac{p_i^2 \theta^2}{6}\right] C_{i*}^{k+1}(x) &= \left[1 + \dfrac{p_i \theta}{3}  \right] C^k_{i-1}(x) \\
p_i &= m_i \left(\log \dfrac{\nu_{i,R}-1}{\nu_{i,R}}\right)\fp{}{x}, \quad i=0,...,M \nn
\end{align}

To preserve the third order of approximation in time the third equation in the Eq.~(\ref{finAlpha3}) should now be solved at $m=0,1,2,3$ and then cubic interpolation to the actual value $m_i$ will give the final solution. This scheme increases the total computational time by about 10\%, however the accuracy in time increases to $O(M \theta^3)$.

\section{Conclusion}
From the numerical point of view the proposed approach has an advantage as compared with the methods mentioned in the Introduction. Indeed, first we managed to reduce the original evolutionary integral equation to a pure differential equation. Second, this equation could be formally solved analytically. To compute the operator exponent we applied a Pad\'e approximation technique. This eventually allowed us to derive finite difference equations which approximate the original solution with the necessary order. This equations could be solved at the same grid as the diffusion part of the original PIDE thus eliminating problems inherent to the FFT methods. In addition, despite the original integral term is non-local, the rhs matrix $\mathcal{D}$ of the system of linear equations obtained by applying our approach is a band matrix in case of integer $\alpha_{R,L}$, i.e. it corresponds to a local approximation of the option price. Also we demonstrated that at $\alpha < 0$ the complexity of our algorithm is much lower than that of the FFT while the accuracy is much better.

The complexity of the solution at $\alpha =1$ is higher than that of the FFT. This in part is compensated by few factors: our algorithm provides the second
order approximation in both space and time, and it does not require to re-interpolate the FFT results to the FD grid which was previously used to find solution for the diffusion part of the original PIDE.

Using this technique the solution at $2 > \alpha > 1$ could be obtained by using extrapolation given the solution at $\alpha = 1,0,-1$.

It is interesting to know what are real values of $\alpha$. In \cite{Bu2007} the author used to calibrate the CGMY model to S\&P 500 historical call option prices. The market prices were chosen from June 2007 to December 2008. The strike is from 1300 to 2000 with the increment of 25 from 1300 to 1700 and the
increment 100 from 1700 to 2000. The index closed price is 1536.34. The found CGMY parameters were CGMY $C=0.0156, G=0.0767, M=7.5500, Y=1.2996$, i.e. $\alpha = 1.3$. In \cite{CarrGemanMadanYor2005} the option prices of S\&P 500 were also calibrated using CGMY model which gave the values of $\alpha$ in the range (-0.39,-0.42).



\newpage
\addcontentsline{toc}{section}{References}
\newcommand{\noopsort}[1]{} \newcommand{\printfirst}[2]{#1}
  \newcommand{\singleletter}[1]{#1} \newcommand{\switchargs}[2]{#2#1}

\end{document}